\documentclass[journal,draftcls,onecolumn,12pt,twoside]{IEEEtranTCOM}
\renewcommand{\baselinestretch}{1.64}
\usepackage{moreverb}
\usepackage{epsfig}
\usepackage{amsmath,amssymb,amsthm,mathrsfs,amsfonts,dsfont}
\usepackage{adjustbox,lipsum}
\usepackage{algorithm,algorithmic}
\usepackage{amsfonts}
\usepackage{epsfig}
\usepackage{amssymb}
\usepackage{amsmath}
\usepackage{amsthm}
\usepackage{multirow}
\usepackage{setspace}
\usepackage{rotating}
\usepackage{graphicx}
\usepackage{tabularx}
\usepackage{array}
\usepackage{anyfontsize}
\usepackage{color,soul}
\usepackage{graphicx,dblfloatfix}
\usepackage{epstopdf}
\usepackage{blindtext}
\usepackage{amsmath}
\usepackage{amsthm,amssymb,amsmath,bm}
\usepackage{subfigure}
\usepackage{amsfonts}
\usepackage{epsfig}
\usepackage{amssymb}
\usepackage{amsmath}
\usepackage{cite}
\usepackage{graphicx}
\usepackage{fancyhdr}
\usepackage[font={small}]{caption}
\usepackage{tabularx}
\usepackage{tcolorbox}
\usepackage{cite}
\usepackage{setspace}

\newtheorem{theorem}{Theorem}
\newtheorem{lemm}{Lemma}

\newtheorem{corollary}{Corollary}

\newtheorem{rem}{Remark}
\newtheorem{Pro}{Proposition}

\allowdisplaybreaks
\begin{document}

\title{Moment Generating Function of Age of Information in Multi-Source M/G/1/1 Queueing Systems
\thanks{Mohammad Moltafet and Markus Leinonen are with the Centre
for Wireless Communications--Radio Technologies, University of Oulu,
90014 Oulu, Finland (e-mail: mohammad.moltafet@oulu.fi; markus.leinonen@oulu.fi), and
Marian Codreanu is with the Department of Science and Technology, Link\"{o}ping University, Sweden (e-mail: marian.codreanu@liu.se).
}
\author{
Mohammad~Moltafet, Markus~Leinonen, and Marian~Codreanu
}
}
\maketitle

\vspace{-15mm}
\begin{abstract}
We consider a multi-source status update system, where each source generates status update packets according to a Poisson process which are then served according to a generally distributed service time. For this multi-source M/G/1/1 queueing model, we introduce a source-aware preemptive packet management policy and derive the moment generating functions (MGFs) of the age of information (AoI) and peak AoI of each source. According to the policy, an arriving fresh packet preempts the possible packet of the same source in the system. Furthermore, we derive the MGFs of the AoI and peak AoI for the source-agnostic preemptive and non-preemptive policy, for which only the average AoI and peak AoI have been derived earlier. Finally, we use the MGFs to derive the average AoI and peak AoI in a two-source M/G/1/1 queueing model under each policy. Numerical results show the effect of the service time distribution parameters on the average AoI: for a given service rate, when the tail of the service time  distribution is sufficiently heavy, the source-agnostic preemptive policy is the best policy, whereas for a sufficiently light tailed distribution, the non-preemptive policy is the best policy. The results also highlight the importance of higher moments of the AoI.


\end{abstract}

\sloppy 

\section{Introduction}\label{Introduction}
Timely delivery of the status updates of various real-world physical processes plays a critical role in enabling the time-critical Internet of Things (IoT) applications.	The age of information (AoI)  was first introduced in the seminal work \cite{5984917} as a destination-centric metric to measure the information freshness in status update systems.
A status update packet contains the measured value of a monitored process and a time stamp representing the time at which the sample was generated. Due to wireless channel access, channel errors, fading, etc. communicating a status update packet through the network experiences a random delay. If at a time instant $t$, the most recently received status update packet contains the time stamp $U(t)$, AoI is defined
as the random process $\Delta(t)=t-U(t)$.
Thus, the AoI measures for each source node the time elapsed since the last received status update packet was generated at the source node.

The first queueing theoretic work on AoI is \cite{6195689} where the authors derived the average AoI for M/M/1, D/M/1, and M/D/1 first-come first-served (FCFS) queueing models. 
In \cite{6875100}, the authors proposed peak AoI as an alternative  metric to evaluate the information freshness.  	The work in \cite{6284003} was the first to investigate the AoI in a multi-source setup in which the authors derived an approximate expression for the average AoI
in a multi-source M/M/1 FCFS queueing model.

It has been shown that an appropriate packet management policy -- in the waiting queue or/and server -- has a great potential to improve the information freshness in status update systems \cite{6310931,7415972}. The average  AoI for an  M/M/1 last-come first-served  (LCFS) queueing model with preemption was analyzed in \cite{6310931}. The average AoI and average peak AoI for three packet management policies named M/M/1/1, M/M/1/2, and M/M/1/$2^*$ were derived in \cite{7415972}.  
The seminal work \cite{8469047} introduced the stochastic hybrid systems (SHS) technique to calculate the average AoI. In \cite{9103131}, the authors extended the SHS analysis to calculate the moment generating function (MGF) of the AoI.  
The SHS technique has been used to analyze the AoI in various queueing models \cite{8437591,8406966,8437907,9013935,9048914,9252168,9162681,Moltafet2020mgf,9611498,Moltafet_thesis}.
The authors of \cite{8437591} considered  a multi-source  queueing model in which sources have different priorities and derived the average AoI for two priority based packet management policies. In  \cite{8406966}, the author derived  the average AoI for a single-source status update system in which the updates follow a route through a series of network nodes where each node has an LCFS queue that supports preemption in service.   
The work  \cite{8437907} derived the average AoI in a single-source  queueing model with multiple servers with preemption in service. 
In  \cite{9013935}, the authors derived the average AoI in a multi-source LCFS queueing model with multiple servers that employ \textit{source-agnostic} preemption in service. According to the source-agnostic preemptive policy, the packets of different sources can preempt each other.
The work in \cite{9048914} derived the average AoI in a multi-source system under a \textit{source-aware} preemptive packet management policy and packet delivery errors.  According to the source-aware preemptive packet management policy, when a packet arrives, the possible packet of the \textit{same source} in the system is replaced by the fresh packet.
The authors of   \cite{9252168,9162681}  derived the average AoI for a multi-source M/M/1 queueing model under various preemptive and non-preemptive packet management policies. In \cite{Moltafet2020mgf}, the authors derived the MGF of the AoI for a multi-source M/M/1 queueing model under various packet management policies.
The authors of \cite{9611498} assumed  that the status update packets received at the sink need further processing before being used and derived the MGF of the AoI for such a two-server tandem queueing system. 


Besides exponentially distributed service time and Poisson process arrivals,  AoI has also been studied under various arrival processes and service time distributions in both single-source and multi-source systems. In \cite{9099557}, the authors derived various approximations for the average AoI in a multi-source M/G/1 FCFS queueing model. The work in \cite{9119460} derived the distribution of the AoI and peak AoI for the single-source PH/PH/1/1 and M/PH/1/2 queueing models. 
The authors of \cite{8406909} analyzed the AoI in a single-source D/G/1 FCFS  queueing model. 
The authors of \cite{8006504}  derived a closed-form expression for the average
AoI of a single-source  M/G/1/1 preemptive queueing model with hybrid automatic repeat request.
The stationary distributions of the AoI and peak AoI of single-source  M/G/1/1 and G/M/1/1 queueing models were derived in \cite{8006592}.
In \cite{8820073}, the authors derived a general formula for the stationary distribution of the AoI in  single-source single-server queueing systems. 
The work in \cite{7541764} considered a single-source LCFS queueing model where the packets arrive according to a Poisson process and the service time follows a gamma distribution. They derived the average AoI and average peak AoI for two packet management policies: LCFS  with the source-agnostic preemptive and non-preemptive. 
According to the non-preemptive policy, when the server is busy any arriving packet is blocked and cleared.
The work in \cite{soysal2019age,9048933} derived the average AoI expression for a single-source G/G/1/1 queueing model under two packet management policies. 
The authors of \cite{7282742} considered a multi-source M/G/1 queueing system and optimized the arrival rates of each source to minimize the peak AoI. 
The average AoI and average peak AoI for a multi-source M/G/1/1 queueing model under the source-agnostic preemption policy were derived in \cite{8406928}.  	
In \cite{8886357}, the authors derived the average AoI for a queueing system with two classes of Poisson arrivals with different priorities under a general service time distribution. They assumed that the system can contain at most one packet and a newly arriving packet replaces the possible currently-in-service packet with the same or lower priority. 
The average AoI and average peak AoI for a multi-source M/G/1/1 queueing model under the source-agnostic non-preemptive policy were derived in \cite{9500775}.

In this work, we consider a multi-source M/G/1/1 queueing system and derive the MGFs of the AoI and peak AoI under three packet management policies, namely, i) source-aware  preemptive policy, ii) source-agnostic preemptive policy \cite{8406928}, and iii) non-preemptive policy \cite{9500775}.  
The capacity of the system is one packet (i.e., there is no waiting buffer).
According to the source-aware preemptive policy, when a packet arrives, the possible packet of \emph{the same source} in the system is replaced by the fresh packet. According to the source-agnostic preemptive policy, a new arriving packet preempts the possible packet in the system regardless of its source index.
According to the non-preemptive  policy,  when the server is busy, any arriving packet is blocked and cleared.
By using the MGFs of the AoI and peak AoI, the average AoI and average peak AoI in a two-source M/G/1/1 queueing system under the three policies are derived. 
The numerical results  show that, depending on the system parameters, the proposed source-aware preemptive packet management policy can outperform the source-agnostic preemptive and non-preemptive policy proposed in \cite{8406928} and \cite{9500775}, respectively, from the perspective of average AoI. In addition, they show the importance of higher moments of the AoI by investigating the standard deviation of the AoI under each policy.

\subsection{Contributions}
The main contributions of this paper are summarized as follows:
\begin{itemize}
\item We introduce a source-aware preemptive packet management policy for a multi-source M/G/1/1 queueing system and derive the MGFs of the AoI and peak AoI under the policy. 
\item 
As an extension of \cite{8406928} and \cite{9500775}, where only the average AoI and peak AoI were derived, we derive the MGFs of the AoI and peak AoI under the source-agnostic preemptive and non-preemptive packet management policies.
\item  By using the MGFs of the AoI and peak AoI, we derive the average AoI and average peak AoI in a two-source M/G/1/1 queueing system under the source-aware preemptive, source-agnostic preemptive, and non-preemptive policies.
\item We numerically investigate the standard deviation of the AoI under the policies and show that the average AoI is not sufficient to rigorously evaluate the information freshness of a status update system for a given packet management policy. 
\item  The numerical results  show that  the average AoI performance of the packet management policies depends on the service time distribution parameters: for a given service rate, when the tail of the service time  distribution is heavy enough,  the source-agnostic preemptive policy is the best policy, and when the tail of the distribution is light enough,  the non-preemptive policy is the best one.

\end{itemize}

\subsection{Organization}
The paper is organized as follows. The system model and summary of the main results are presented in Section \ref{System Model and Summary of Results}. Calculation of the MGFs of the AoI and  peak AoI is presented in Section \ref{Calculation of the MGF of the AoI under the packet management policies}. 
Numerical results are presented in Section \ref{Numerical Results}. Finally, concluding remarks are made in Section \ref{Conclusions}.

\section{System Model and Main Results}\label{System Model and Summary of Results}
We  consider a status update system consisting of  a set of independent sources denoted by $\mathcal{C}=\{1,\dots,C\}$,  one server, and one sink, as depicted in Fig.~\ref{Model}.
Each source is assigned to send status information about a random process to the sink. Status updates are transmitted as packets, containing the measured value of the monitored process and a time stamp representing the time when the sample was generated. We assume that  the packets  of source  $c\in\mathcal{C}$  are generated according to the   Poisson process with the rate $\lambda_c$. Since packets of each source are generated according to a Poisson process and the sources are independent, the packet generation in the system follows
the Poisson process with rate $\lambda=\sum_{c'\in\mathcal{C}}\lambda_c'$.
The server serves the packets according to a generally distributed service time with rate $\mu$. We assume that the service times of  packets are independent
and identically distributed (i.i.d.) random variables following
a general distribution. 
Finally, we consider that the capacity of the system is one (i.e., there is no waiting buffer) and thus, the considered setup is referred to as a multi-source M/G/1/1 queueing system. 
\begin{figure}
\centering
\includegraphics[width=.45\linewidth,trim = 0mm 0mm 0mm 0mm,clip]{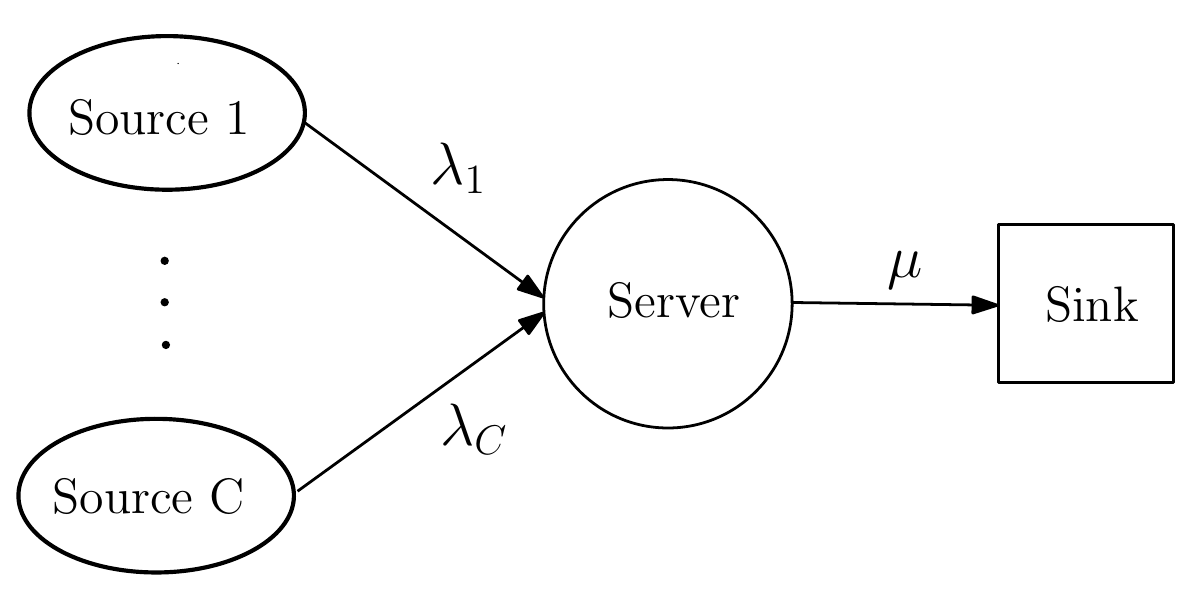}
\caption{The considered multi-source M/G/1/1 queueing system.}  
\label{Model}
\vspace{-10mm}		
\end{figure}

\subsection{Packet Management Policies}
In this paper, we study the following three packet management policies:

\textbf{Source-Aware Preemptive Policy:} According to this policy, a new arriving packet preempts the possible packet of the same source in the system. Whenever the new arriving packet finds a packet of another source under service, the arriving packet is blocked and cleared. 

\textbf{Source-Agnostic Preemptive Policy \cite{8406928}:} According to this policy, a new arriving packet preempts the possible packet in the system regardless of its source index.

\textbf{Non-Preemptive Policy \cite{9500775}:} According to this policy, when the server is busy at the arrival instant of a packet, the arriving packet is blocked and cleared. 




\subsection{AoI Definition}
For each source, the AoI at the sink is defined as the time elapsed since the last  successfully received packet was generated. 
Formally, let $t_{c,i}$ denote the time instant at which the $i$th delivered status update packet of source $c$ was generated, and let $t'_{c,i}$ denote the time instant at which this packet  arrives at the sink. Let $\bar t_{c,i}$ denote the generation time of the $i$th packet of source $c$ that does not complete service because of the packet management policy (i.e., the packet is either preempted by another packet or it is blocked and cleared). An example of the evolution of the AoI in a two-source system under the source-aware preemptive packet management policy is shown in Fig.~\ref{Sawtooth}.

\begin{figure}
\centering
\includegraphics[width=.9\linewidth,trim = 0mm 0mm 0mm 0mm,clip]{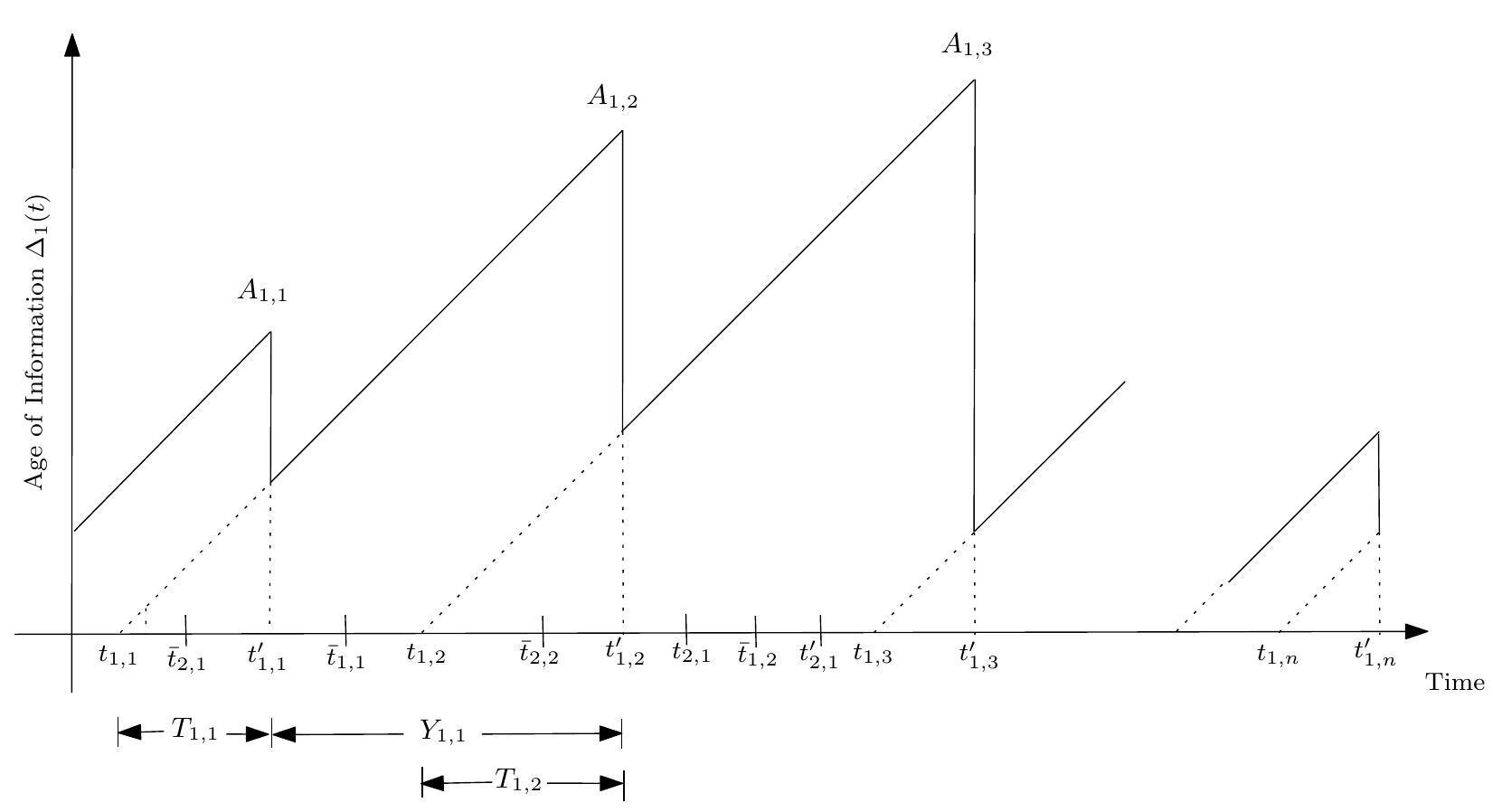}\vspace{-3mm}
\caption{An example of the evolution of the AoI of source 1 in a two-source system under the source-aware preemptive packet management policy. The first packet of source 1 is generated at time instant $t_{1,1}$ and this packet is delivered to the sink at time instant $t'_{1,1}$. The first packet of source 2 arrives at the system at time instant $\bar t_{2,1}$; however, because the server is serving a source 1 packet, the arrived source 2 packet is blocked and cleared. At time instant $\bar t_{1,1}$, a source 1 packet arrives at the empty system and starts the service; however, this packet is replaced by the new packet of source 1 arriving at time instant $t_{1,2}$.}  
\label{Sawtooth}
\vspace{-10mm}			
\end{figure}

At a time instant $\tau$,  the  index of the most recently received packet of source $c$ is given by
\begin{equation}\label{mnb00}
N_c(\tau)=\max\{i'\mid t'_{c,i'}\le \tau\},
\end{equation}
and the time stamp of the most recently received packet of source $c$ is
$
U_c(\tau)=t_{c,N_c(\tau)}.
$
The AoI of source $c$ at the sink is defined as the random process
$
{\delta_{c}(t)=t-U_c(t).}
$

Let the random variable
\begin{equation}\label{mn120}
Y_{c,i}=t'_{c,i+1}-t'_{c,i}
\end{equation}
represent the $i$th interdeparture time of source $c$, i.e., the time elapsed between the departures of $i$th and $i+1$th (delivered) packets from source $c$. From here onwards, we refer to the $i$th delivered packet from source $c$ simply as ``packet $c,i$''. Moreover, let the random variable
\begin{equation}\label{mn1201}
T_{c,i}=t'_{c,i}-t_{c,i}
\end{equation}
represent the system time of packet $c,i$, i.e., the duration this (delivered) packet spends in the system. 

One of the most commonly used metrics for evaluating the
AoI of a source at the sink is the peak AoI \cite{6875100}.
The peak AoI of  source $c$ at the sink  is defined as the value of the AoI immediately  before  receiving  an update packet.
Accordingly, the peak AoI concerning the $i$th successfully received packet of source $c$,  denoted by $A_{c,i}$ (see  Fig.\ \ref{Sawtooth} for the source-aware preemptive policy), is given by
\begin{align}\label{AoI.eq}
A_{c,i}= Y_{c,i-1}+T_{c,i-1}.
\end{align}

We assume that the considered status update system is stationary so that $T_{c,i}=^{\mathrm{st}}T_c$,  $Y_{c,i}=^{\mathrm{st}}Y_c$, and $A_{c,i}=^{\mathrm{st}}A_c,\forall i$, where ${=^{\mathrm{st}}}$ means stochastically identical (i.e., they have an identical marginal distribution). We further assume that the AoI process for each source is ergodic. 
%

Next, the main results of the paper are presented. The results are valid for any service time distribution under the three packet management policies. 

\subsection{Summary of the Main Results}\label{Summary of the Main Results}
The MGFs of the AoI and peak AoI of source $c$ in a multi-source M/G/1/1 queueing system under each of the three packet management policies are given by the following three theorems; the proofs of Theorems \ref{T_source-aware}, \ref{T_preemptive}, and  \ref{T_blocking} are presented in Section~\ref{Calculation of the MGF of the AoI under the packet management policies}. 

Let $S$ be the random variable representing the service time of any packet in the system.

\begin{theorem}\label{T_source-aware}
The MGFs of the AoI and peak AoI of source $c$ under the  source-aware preemptive packet management policy, denoted by $\bar{M}_{\delta_c}(s)$ and $\bar{M}_{A_c}(s)$, respectively, are given as
\begin{align}
{\bar{M}_{\delta_c}(s)}=\dfrac{M_S(s-\lambda_c)(\bar{M}_{Y_c}(s)-1)}{sL_{\lambda_c}\bar{M}'_{Y_c}(0)},
\end{align}
\begin{align}
\bar{M}_{A_c}(s)=\dfrac{M_S(s-\lambda_c)\bar{M}_{Y_c}(s)}{L_{\lambda_c}},
\end{align}
where $L_{\lambda_c}=\mathbb{E}[e^{-\lambda_cS}]$, $M_S(s-\lambda_c)=\mathbb{E}[e^{(s-\lambda_c)S}]$ is the MGF of the service time $S$ at ${s-\lambda_c}$, $\bar{M}_{Y_c}(s)$ is the MGF of the interdeparture time $Y_c$ under the policy, which is given as 
\begin{align}\label{mgfinterdeparty0}
\bar{M}_{Y_c}(s)=\dfrac{a_cM_S(s-\lambda_c)}{(1-a'_c)\left(1-\sum_{c'\in\mathcal{C}\setminus\{c\}}\dfrac{a_{c'}M_S(s-\lambda_{c'})}{1-a'_{c'}}\right)},
\end{align}
where $a_c=\dfrac{\lambda_c}{\lambda-s}$ and $a'_c=\dfrac{\lambda_c(1-M_S(s-\lambda_c))}{\lambda_c-s}$, and $\bar{M}'_{Y_c}(0)$ is the first derivative of the MGF of $Y_c$ under the policy, evaluated at $s=0$, i.e.,
$$
\bar{M}'_{Y_c}(0)=\dfrac{\mathrm{d}(\bar{M}_{Y_c}(s))}{\mathrm{d}s}\Big|_{s=0}.
$$
\end{theorem}


\begin{theorem}\label{T_preemptive}
The MGFs of the AoI and peak AoI of source $c$ under the  source-agnostic preemptive packet management policy, denoted by $\hat{M}_{\delta_{c}}(s)$ and $\hat{M}_{A_c}(s)$, respectively, are given as
\begin{align}
\hat{M}_{\delta_{c}}(s)=\dfrac{\lambda_cM_S(s-\lambda)}{\lambda_cM_S(s-\lambda)-s},
\end{align}
\begin{align}
\hat{M}_{A_c}(s)=\dfrac{\lambda_cM^2_S(s-\lambda)}{L_{\lambda}(\lambda_cM_S(s-\lambda)-s)}.	
\end{align}
\end{theorem}

\begin{theorem}\label{T_blocking}
The MGFs of the AoI and peak AoI of source $c$ under the  non-preemptive packet management policy, denoted by $\tilde{M}_{\delta_{c}}(s)$ and $\tilde{M}_{A_c}(s)$, respectively, are given as
\begin{align}
\tilde{M}_{\delta_{c}}(s)=\dfrac{\lambda_cM_S(s)(s+\lambda(M_S(s)-1))}{s(\lambda/\mu+1)(\lambda-s-(\lambda-\lambda_c)M_S(s))},	
\end{align}
\begin{align}
\tilde{M}_{A_c}(s)=\dfrac{\lambda_cM^2_S(s)}{\lambda-s-(\lambda-\lambda_c)M_S(s)}.
\end{align}
\end{theorem}

\begin{rem}\label{rem1MGFage}\normalfont
The $m$th moment of the AoI (peak AoI) is derived by calculating the $m$th derivative of the MGF of the AoI (peak AoI) when $s\rightarrow 0$. For instance, considering the source-aware preemptive packet management policy, the $m$th moment of the AoI and peak AoI  are given as
\begin{align}
\bar{\Delta}_c^m=\dfrac{\mathrm{d}^m(\bar{M}_{\delta_{c}}(s))}{\mathrm{d}s^m}\Big|_{s=0},~~~~
\bar{A}_c^m=\dfrac{\mathrm{d}^m(\bar{M}_{A_c}(s))}{\mathrm{d}s^m}\Big|_{s=0}.
\end{align}
\end{rem}

In the next three corollaries, by using  Theorems \ref{T_source-aware}, \ref{T_preemptive}, and  \ref{T_blocking} and Remark \ref{rem1MGFage}, we derive the average AoI and average peak AoI of source 1 in a \textit{two-source} M/G/1/1 queueing model under each of the three  packet management polices. 
\begin{corollary}\label{agemg11theorem}
The  average AoI and average peak AoI of source 1 in a two-source M/G/1/1 queueing model under the source-aware preemptive packet management  policy are given as
\begin{align}\nonumber
\bar{\Delta}_1&= \dfrac{L_{\lambda_1}^2(\lambda(1-L_{\lambda_2})-\lambda_1\lambda_2L_{\lambda_2}')+L_{\lambda_2}^2(\lambda_2(1-L_{\lambda_1}))-\Psi}{\lambda_1\lambda_2L_{\lambda_1}L_{\lambda_2}(L_{\lambda_1}+L_{\lambda_2}-L_{\lambda_1}L_{\lambda_2})},
\end{align}
\begin{align}\nonumber
\bar{A}_1&= \dfrac{L_{\lambda_1}+L_{\lambda_2}-L_{\lambda_1}L_{\lambda_2}+\lambda_1L_{\lambda_2}L_{\lambda_1}'}{\lambda_1L_{\lambda_1}L_{\lambda_2}},
\end{align}
where  
$\Psi=\lambda_1\lambda_2L_{\lambda_1}L'_{\lambda_1}+L_{\lambda_1}L_{\lambda_2}\lambda_2(1+\lambda_1L_{\lambda_1}')$ and $ L'_{\lambda_c}=\mathbb{E}[Se^{-\lambda_cS}] $. 	
\end{corollary}
\begin{rem}
The average AoI under the source-aware preemptive policy, presented in Corollary~\ref{agemg11theorem}, generalizes the existing results in \cite{8006504} and
\cite{9048914}. Specifically, when confining to a single-source case by letting $ \lambda_2 \rightarrow 0 $, the average
AoI becomes equal to that of the single-source M/G/1/1 queueing model with preemption derived in \cite{8006504}. Moreover, when we consider an exponentially distributed service time, the average
AoI expression coincides with that of the multi-source M/M/1/1 queueing model with preemption derived in \cite{9048914}.
\end{rem}

\begin{corollary}\label{agemg11theorempree}
The  average AoI and average peak AoI of source 1 in a two-source M/G/1/1 queueing model under the source-agnostic preemptive  packet management  policy are given as
\begin{align}\nonumber
\hat{\Delta}_1= \dfrac{1}{\lambda_1 L_{\lambda}},
~~~~~~~\hat{A}_{1}(s)=\dfrac{1+\lambda_1L'_{\lambda}}{\lambda_1 L_{\lambda}}.
\end{align}

\end{corollary}

\begin{corollary}\label{agemg11theoremblock}
The  average AoI and average peak AoI of source 1 in a two-source M/G/1/1 queueing model under the non-preemptive packet management  policy are given as
\begin{align}\nonumber
\tilde{\Delta}_1= \dfrac{\lambda+\mu}{\lambda_1\mu}+\dfrac{\lambda\mu\mathbb{E}[S^2]}{2(\lambda+\mu)}.
~~~~~~~\tilde{A}_{1}(s)=\dfrac{2\lambda_1+\lambda_2+\mu}{\lambda_1\mu}.
\end{align}
\end{corollary}


It is worth noting that the results in Corollaries \ref{agemg11theorempree} and \ref{agemg11theoremblock} were previously derived in \cite{8406928} and \cite{9500775}, respectively (without deriving the MGFs). Thus, our derived MGF expressions generalize the results in \cite{8406928,9500775}, i.e., besides the first moment, they can readily be used to derive higher moments of the AoI and peak AoI.

\section{Derivation of the MGFs of the AoI and Peak AoI }\label{Calculation of the MGF of the AoI under the packet management policies}

In this section, we prove Theorems \ref{T_source-aware}, \ref{T_preemptive}, and  \ref{T_blocking}.
To prove  the theorems, we first provide Lemma~\ref{lemmsmgfage} which presents the MGF of the AoI of source $c$ in the considered multi-source M/G/1/1 queueing model as a function of the MGFs of the system time of  source $c$, $T_c$, and interdeparture time of  source $c$, $Y_c$. It is worth noting that the presented MGF expression is valid for the source-aware preemptive, source-agnostic preemptive, and non-preemptive  packet management policies.
\begin{lemm}\label{lemmsmgfage}
The MGFs of the AoI and peak AoI of source $c$ in a multi-source M/G/1/1 queueing model under the source-aware preemptive, source-agnostic preemptive, and non-preemptive packet management policies, denoted by ${M}_{\delta_{c}}(s)$ and ${M}_{A_c}(s)$, respectively, can be expressed as
\begin{align}\label{MGFofagegeneral}
M_{\delta_{c}}(s)=\dfrac{M_{T_c}(s)(M_{Y_c}(s)-1)}{s\mathbb{E}[Y_c]},
\end{align}
\begin{align}\label{MGFpeak}
M_{A_c}(s)=M_{T_c}(s)M_{Y_c}(s),
\end{align}
where $M_{T_c}(s)$ is the  MGF of the system  time of a delivered packet of source $c$ and $M_{Y_c}(s)$ is the  MGF of the interdeparture time of source $c$; these MGFs need to be determined specific to the packet management policy.
\end{lemm}
\begin{proof} 
Let an \textit{informative packet} refer to a successfully delivered packet from source $c$; otherwise, the packet is termed \textit{non-informative}. 
{By invoking the result in \cite[Theorem~10]{8820073} and applying it in our considered multi-source M/G/1/1 queueing system, if the following three conditions are satisfied} 
\begin{enumerate}
\item The  arrival rate of the informative
packets is positive and finite;
\item The system is stable;
\item The marked point process $\{(t'_{c,i},T_{c,i})\}_{i=1,2,\ldots}$ is ergodic;
%
\end{enumerate}
%
%
%
%
then, the Laplace transform of the AoI of source $c$, $L_{\delta_c}(s)$, is given  as 
\begin{align} \label{Laplace of age}
L_{\delta_c}(s)=\bar{\lambda}_c\dfrac{L_{T_c}(s)-L_{A_c}(s)}{s},
\end{align}
where $\bar{\lambda}_c$ is the arrival rate of informative packets,
$ L_{T_c}(s) $ is the Laplace transform of the system time of any delivered packet from source $c$, and $L_{A_c}(s)$ is the Laplace transform of the peak AoI of source $c$. 
Next, we verify the conditions for the multi-source M/G/1/1 queueing model under the three packet management policies.

Condition 1: Since the packets of source $c$, both informative and non-informative, arrive according to the Poisson process with rate $\lambda_c$, the mean arrival rate of informative packets is finite. The assumption  that the arrival rate of informative packets is positive, i.e., $\bar\lambda_c\ne 0$, is a reasonable assumption for any well-behaving status update system, since otherwise the AoI would go to infinity. 

Condition 2: Since the capacity of the considered system is one packet, i.e., there are no waiting rooms in the system, the system is stable under the three packet management policies. 
Moreover, since an informative packet refers to a successfully delivered packet from source $c$ and the system is stable,  
the mean arrival rate of informative packets of source $c$, $\bar\lambda_c$, is equal to the mean departure rate of the packets which is calculated by $\lim_{\tau\to\infty}\dfrac{N_c(\tau)}{\tau}$, where $N_c(\tau)$ is the number of delivered packets until time $\tau$.

Condition 3: 
If we ignore the non-informative packets and just observe the informative packets, the system can be considered as an  FCFS queueing model serving (only) the informative packets. In addition,  since the system is stable under the three policies, according to {\cite[Sect.~X,~ Proposition~ 1.3]{asmussen2008applied}}, the system times of informative packets, $\{T_{c,i}\}_{i=1,2,\ldots}$, form a regenerative process with finite mean regeneration time.
Therefore, it can be verified that
$\{(t'_{c,i},T_{c,i})\}_{i=1,2,\cdots}$ is mixing {\cite[Page~49]{baccelli2003elements}}, and consequently, it is ergodic. 


The 
Laplace transform and the MGF of the AoI are interrelated as  
\begin{align}\label{Moment-Laplace}
M_{\delta_{c}}(s)=\mathbb{E}[e^{s\delta_c}]&=L_{\delta_c}(-s)\stackrel{(a)}{=}\bar{\lambda}_c\dfrac{L_{A_c}(-s)-L_{T_c}(-s)}{s},
\end{align}
where $(a)$ follows from \eqref{Laplace of age}.
Similarly, for the MGF of the peak AoI of source $c$, $M_{A_c}(s)$, we have $M_{A_c}(s)=L_{A_c}(-s)$; 
and for the MGF of the system time of a delivered packet of source $c$, we have $M_{T_c}(s)=L_{T_c}(-s)$. Accordingly, \eqref{Moment-Laplace} can be written as  
\begin{align}\label{Moment-Laplace2}
M_{\delta_{c}}(s)=\bar{\lambda}_c\dfrac{M_{A_c}(s)-M_{T_c}(s)}{s}.
\end{align}
As shown in \eqref{AoI.eq}, the peak AoI of source $c$ can be presented as a summation of two independent random variables, $Y_{c,i-1}$ and $T_{c,i-1}$. Using the basic features of an MGF,  the MGF of the peak AoI, $M_{A_c}(s)$, is given as the product of the MGFs of random variables $Y_{c,i-1}$ and $T_{c,i-1}$, i.e.,
\begin{align}\label{MGFpeak2}
M_{A_c}(s)=M_{T_c}(s)M_{Y_c}(s).
\end{align}

Since interdeparture times between consecutive  packets of source $c$ under each of the three policies are i.i.d., the number of delivered packets until time $\tau$, $ N_c(\tau) $, forms a renewal process. Thus, we have
\begin{align}\label{effective arrival}
\bar\lambda_c=\lim_{\tau\to\infty}\dfrac{N_c(\tau)}{\tau}=\dfrac{1}{\mathbb{E}[Y_c]}.
\end{align}

Substituting  \eqref{MGFpeak}   and \eqref{effective arrival} into \eqref{Moment-Laplace2} completes the proof of Lemma~\ref{lemmsmgfage}.
\end{proof}

According to Lemma \ref{lemmsmgfage}, the main challenge in calculating the MGFs  of the AoI (see \eqref{MGFofagegeneral}) and peak AoI  (see \eqref{MGFpeak}) under each packet management policy reduces to  deriving the  MGF of the system  time of source $c$, $M_{T_c}(s)$, and the MGF of the interdeparture time of source $c$, $M_{Y_c}(s)$. Note that when we have $M_{Y_c}(s)$, we can easily derive $\mathbb{E}[Y_c]$ (as will be shown in  Remark~\ref{rem1MGFage}).

Next, we will derive the MGFs of the AoI and peak AoI under the source-aware preemptive, source-agnostic preemptive, and non-preemptive packet management policies.

\subsection{MGFs of AoI and Peak AoI Under the Source-Aware Preemptive Packet Management Policy}
To derive the MGF of the system  time of source $c$, we first derive the probability density function (PDF) of the system time, $f_{T_c}(t)$, which is given by the following lemma. 
\begin{lemm}\label{Lemma1}
The PDF of the system time of source $c$, $f_{T_c}(t)$, is given by
\begin{align}\label{pdfsystemtime}
f_{T_c}(t)=\dfrac{f_S(t)e^{-\lambda_ct}}{L_{\lambda_c}}.
\end{align}
\end{lemm}
\begin{proof}
The system time of a delivered packet from source $c$ is equal to the service time of the packet. 
Let $X_c$ be a random variable representing the interarrival time between two consecutive packets of source $c$. Thus, the distribution of $T_c$ is given by  ${\mathrm{Pr}(T_c>t)=\mathrm{Pr}(S>t\mid S<X_c)}$. Hence,  $f_{T_c}(t)$ is calculated as 
\begin{align}\label{f_t1}
f_{T_c}(t)&=\lim_{\epsilon\rightarrow 0}\dfrac{\mathrm{Pr}(t<T_c<t+\epsilon)}{\epsilon}\\&\nonumber
\stackrel{}{=}\lim_{\epsilon\rightarrow 0}\dfrac{\mathrm{Pr}(t<S<t+\epsilon\mid S<X_c)}{\epsilon}\\&\nonumber
=\lim_{\epsilon\rightarrow 0}\dfrac{\mathrm{Pr}(t<S<t+\epsilon)\mathrm{Pr}(S<X_c\mid t<S<t+\epsilon)}{\epsilon\mathrm{Pr}(S<X_c)}\\&\nonumber
\stackrel{}{=}\dfrac{f_S(t)\mathrm{Pr}(X_c>t)}{\mathrm{Pr}(S<X_c)}\\&\nonumber
\stackrel{(a)}{=}\dfrac{f_S(t)e^{-\lambda_ct}}{L_{\lambda_c}},
\end{align}
where $(a)$ follows from the fact that i) the interarrival times of the source $ c $ packets follow the exponential distribution with parameter $\lambda_c$ and thus, ${\mathrm{Pr}(X_c>t)=1-F_{X_c}(t)=e^{-\lambda_ct}}$, where $F_{X_c}(t)$ is the cumulative distribution function (CDF) of the interarrival time $X_c$ and ii) ${\mathrm{Pr}(S<X_c)}$  is calculated as 
\begin{align}\label{barpc}
\mathrm{Pr}(S<X_c)&=\int_{0}^{\infty}\mathrm{Pr}(S<X_c\mid X_c=t)f_{X_c}(t)\mathrm{d}t\\&\nonumber
=\int_{0}^{\infty}F_s(t)\lambda_ce^{-\lambda_ct}\mathrm{d}t\stackrel{(b)}{=} L_{\lambda_c},
\end{align}
where $F_s(t)$ is the CDF of the service time $S$, and $(b)$ follows from the fact that according to the feature of the Laplace transform, for any function $f(y), y\ge 0$, we have \cite[Sect.~13.5] {rade2013mathematics}:
\begin{align}\label{nbhg01}
L_{\int_{0}^{y}f(b)\mathrm{d}b}(s)=\dfrac{L_{f(y)}(s)}{s},
\end{align}
where $L_{f(y)}(s)$ is the Laplace transform of $f(y)$.
\end{proof}

Using Lemma~\ref{Lemma1}, the MGF of the system time of source $c$, $\bar{M}_{T_c}(s)=\int_{0}^{\infty}e^{st}f_{T_c}(t)\mathrm{d}t$, is given as
\begin{align}\label{mgfsystemtime}
\bar{M}_{T_c}(s)&=\dfrac{1}{L_{\lambda_c}}\int_{0}^{\infty}e^{(s-\lambda_c)t}f_S(t)\mathrm{d}t\\&\nonumber=\dfrac{M_S(s-\lambda_c)}{L_{\lambda_c}}.
\end{align}

The next step is to derive the MGF of the interdeparture time $Y_c$, $\bar{M}_{Y_c}(s)$, which is given by the following proposition.
\begin{Pro}\label{Lemma2}
The MGF of the  interdeparture time of  source $c$, $\bar{M}_{Y_c}(s)$, is given by
\begin{align}\label{mgfinterde1}
\bar{M}_{Y_c}(s)=\dfrac{a_cM_S(s-\lambda_c)}{(1-a'_c)\left(1-\sum_{c'\in\mathcal{C}\setminus\{c\}}\dfrac{a_{c'}M_S(s-\lambda_{c'})}{1-a'_{c'}}\right)},
\end{align}
where   $a'_c$ and $a_c$ were defined below \eqref{mgfinterdeparty0}.
\end{Pro}
\begin{proof}
The MGF of the interdeparture time of source $c$ packets is defined as ${\bar{M}_{Y_c}(s)=\mathbb{E}[e^{sY_c}]}$. To derive $\bar{M}_{Y_c}(s)$, we need to first characterize $Y_c$. To this end, Fig. \ref{Semi-Chain_c} depicts a semi-Markov chain that represents the different system occupancy states (indicated by $q$'s) and their transition probabilities (indicated by $p$'s) in relation to $Y_c$, i.e., the dynamics of the system occupancy of the $C$ different sources' packets in relation to $Y_c$. Thus, the graph captures all the probabilistic queueuing-related events that constitute the interdeparture time $Y_c$, allowing us to derive $Y_c$.

For the graph in Fig.\ \ref{Semi-Chain_c}, the $C+2$ states $\{q_0,q_1,q_2,\ldots,q_{C},q'_0\}$ are explained as follows. When a source $c$ packet is successfully delivered to the sink, the system goes to idle state $q_0$, where it waits for a new arrival from any source. State $q_{c'},~c'\in\mathcal{C}$, indicates that a source $c'$ packet is under service. State $q'_0$ indicates that a packet of source  $c'\in\mathcal{C}_{-c}$ is successfully delivered to the sink and the system becomes empty, where $\mathcal{C}_{-c}=\mathcal{C}\setminus\{c\}$.
From the graph, the interdeparture time $Y_c$ is calculated by characterizing the required time to start from state $q_0$ and return to $q_0$. Let $\bar X_c=\min_{c'\in\mathcal{C}_{-c}} X_{c'}$; then, the transitions between the states are explained in the following:
\begin{enumerate}
\item $ q_0\rightarrow q_{c'},~\forall c'\in\mathcal{C}$: The system is in the idle state $q_0$ and a source $ c' $ packet arrives. This transition happens {if the interarrival time of source $c'$ packet, $X_{c'}$, is shorter than the minimum interarrival time among all the other sources, $\bar X_{c'}$.} Thus, the transition occurs with probability $p_{c'}=\mathrm{Pr}(X_{c'}<\bar X_{c'})$. The sojourn time of the system in state $q_0$ before this transition, denoted by $\eta_{c'}$, has the distribution  ${\mathrm{Pr}(\eta_{c'}>t)=\mathrm{Pr}(X_{c'}>t\mid X_{c'}<\bar X_{c'})}$.


\item $q_{c'}\rightarrow q_{c'},~\forall c'\in\mathcal{C}$: The system is in state $q_{c'}$, i.e., serving a source $c'$ packet, while a new source $c'$ packet arrives and enters the system due to the source-aware preemptive packet management policy. This transition happens with probability $p'_{c'}=\mathrm{Pr}(X_{c'}<S)$. The sojourn time of the system in state $q_{c'}$ before this transition, denoted by $\eta'_{c'}$, has the distribution  ${\mathrm{Pr}(\eta'_{c'}>t)=\mathrm{Pr}(X_{c'}>t\mid X_{c'}<S)}$.

\item $q_c\rightarrow q_0$: The system is in state $q_c$ and the source $c$ packet completes service and is delivered to the sink.  This transition happens with probability $\bar p_c=\mathrm{Pr}(S<X_c)$. The sojourn time of the system in state $q_c$ before this transition, denoted by $\bar{\eta}_{c}$, has the distribution  ${\mathrm{Pr}(\bar \eta_c>t)=\mathrm{Pr}(S>t\mid S<X_c)}$.


\item $q_{c'}\rightarrow q'_0,~\forall c'\in\mathcal{C}_{-c}$: The system is in state $q_{c'},~\forall c'\in\mathcal{C}_{-c}$, and the source $c'$ packet completes service and is delivered to the sink.  This transition happens with probability $\bar p_{c'}=\mathrm{Pr}(S<X_{c'})$. The sojourn time of the system in state $q'_{c'}$ before this transition has the distribution  ${\mathrm{Pr}(\bar \eta_{c'}>t)=\mathrm{Pr}(S>t\mid S<X_{c'})}$.

\item $ q'_0\rightarrow q_{c'},~\forall {c'}\in\mathcal{C}$: This transition is the same as transition  $ q_0\rightarrow q_{c'}$.


\end{enumerate}

\begin{figure}
\centering
\includegraphics[width=.5\linewidth,trim = 0mm 0mm 0mm 0mm,clip]{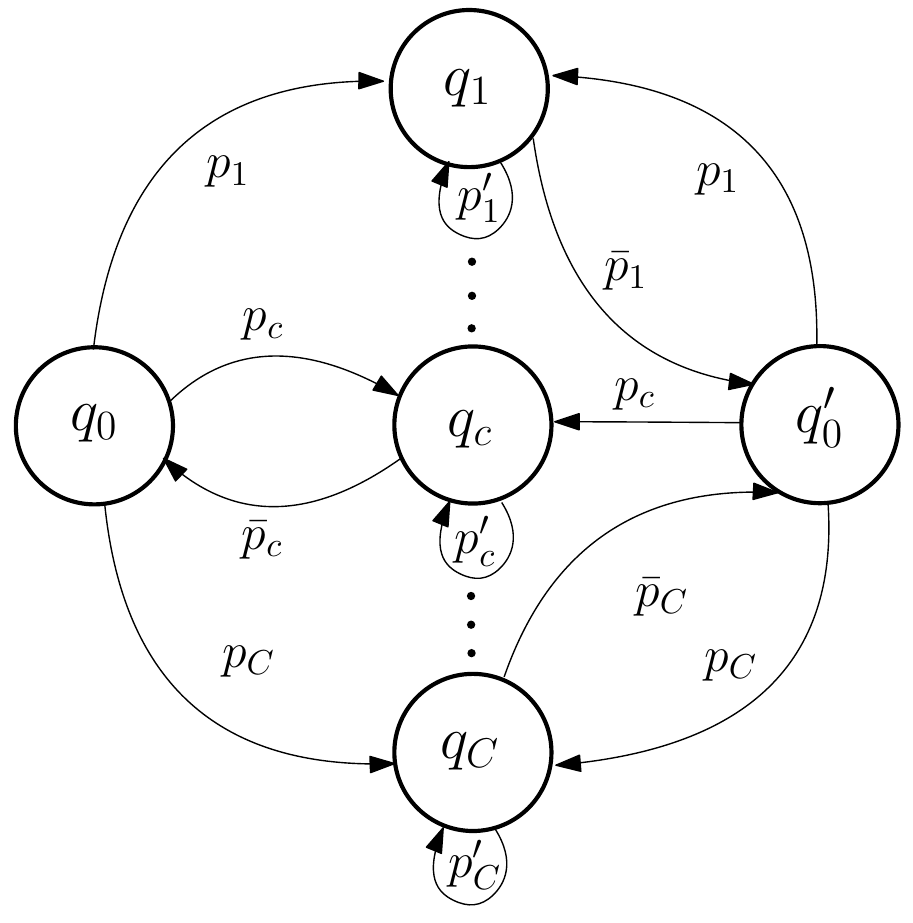}
\caption{The semi-Markov chain corresponding to the interdeparture time of two consecutive packets of source $ c $ under the source-aware preemptive policy, $Y_{c}$. }  
\label{Semi-Chain_c}
\vspace{-10mm}
\end{figure}

Next, we derive the transition probabilities and the sojourn time distributions. 
\begin{lemm}\label{prob01}
The transition probabilities $p_{c'}$, $p'_{c'}$, and  $\bar p_{c'}$ for all ${c'}\in\mathcal{C}$ are given as follows:
\begin{align}
&p_{c'}=\dfrac{\lambda_{c'}}{\lambda},~~~~
\bar p_{c'}=L_{\lambda_{c'}},~~~~
p'_{c'}=1-L_{\lambda_{c'}}.
\end{align}
\end{lemm}
\begin{proof}
Since $\bar X_{c'}$ is the minimum of independent exponentially distributed random variables ${X_{j},~{j}\in\mathcal{C}_{-{c'}}}$, it follows the exponential distribution with parameter $\bar\lambda_{c'}=\sum_{j\in\mathcal{C}_{-{c'}}}\lambda_j$. Thus, we have
\begin{align}\nonumber
p_{c'}&=\mathrm{Pr}(X_{c'}<\bar X_{c'})\\&\nonumber
=\int_{0}^{\infty}\mathrm{Pr}(X_{c'}<\bar X_{c'}\mid \bar X_{c'}=t)f_{\bar X_{c'}}(t)\mathrm{d}t\\&
=\int_{0}^{\infty}(1-e^{-\lambda_{c'}t})\bar\lambda_{c'}e^{-\bar\lambda_{c'}t}\mathrm{d}t=\dfrac{\lambda_{c'}}{\lambda}.
\end{align}

The probability $\bar p_{c'}=\mathrm{Pr}(S<X_{c'})=L_{\lambda_{c'}}$ was derived in \eqref{barpc}. In addition, we have ${p'_{c'}=\mathrm{Pr}(X_{c'}<S)=1-\mathrm{Pr}(X_{c'}>S)=1-\bar p_{c'}=1-L_{\lambda_{c'}}.}$
\end{proof}

\begin{lemm}\label{lemmfaap}
The PDFs of the sojourn time random variables $\eta_{c'}$, $\bar\eta_{c'}$, and $\eta'_{c'}$ for all ${c'}\in\mathcal{C}$ are given as follows:
\begin{align}
&f_{\eta_{c'}}(t)=\lambda e^{-\lambda t},\\&\nonumber 
f_{\bar \eta_{c'}}(t)=\dfrac{f_s(t)e^{-\lambda_{c'}t}}{L_{\lambda_{c'}}},\\&\nonumber
f_{\eta'_{c'}}(t)=\dfrac{\lambda_{c'}e^{-\lambda_{c'}t}(1-F_s(t))}{1-L_{\lambda_{c'}}}.
\end{align}
\end{lemm}

\begin{proof}
We only prove the PDF of the random variable  $\eta_{c'}$; the other PDFs can be derived using the same approach. The PDF of the random variable  $\eta_{c'}$ is given as
\begin{align}\label{f_A}
f_{\eta_{c'}}(t)&=\lim_{\epsilon\rightarrow 0}\dfrac{\mathrm{Pr}(t<\eta_{c'}<t+\epsilon)}{\epsilon}\\&\nonumber
\stackrel{}{=}\lim_{\epsilon\rightarrow 0}\dfrac{\mathrm{Pr}(t<X_{c'}<t+\epsilon\mid X_{c'}<\bar X_{c'})}{\epsilon}\\&\nonumber
=\lim_{\epsilon\rightarrow 0}\dfrac{\mathrm{Pr}(t<X_{c'}<t+\epsilon)\mathrm{Pr}(X_{c'}<\bar X_{c'}\mid t<X_{c'}<t+\epsilon)}{\epsilon\mathrm{Pr}(X_{c'}<\bar X_{c'})}\\&\nonumber
\stackrel{}{=}\dfrac{(1-F_{\bar X_{c'}}(t))f_{X_{c'}}(t)}{\mathrm{Pr}(X_{c'}<\bar X_{c'})}
\stackrel{}{=}\lambda e^{-\lambda t}.
\end{align}
\end{proof}

To reiterate, according to Fig.\ \ref{Semi-Chain_c}, the interdeparture time between two consecutive packets from source $c$ is equal to {the total sojourn time} experienced by the system between starting from $q_0$ and returning to $q_0$. That is, this total sojourn time consists of a summation of the individual sojourn times -- which are specific to each state and its related transitions -- for all possible paths $\{q_0,\ldots,q_0\}$. Thus, random variable $Y_c$ can be characterized by the sojourn time random variables $\eta_{c'}$, $\bar \eta_{c'}$, and $\eta'_{c'}$ for all ${c'}\in\mathcal{C}$, and their numbers of occurrences, which are denoted by $ k_{c'}$, $\bar k_{c'}$, and $k'_{c'}$, respectively. Consequently, $Y_c$ can be presented as 
\begin{align}\label{Y_c_sojourn}
Y_c=\sum_{c'\in\mathcal{C}}k_{c'}\eta_{c'}+\sum_{{c'}\in\mathcal{C}}\bar k_{c'}\bar\eta_{c'}+\sum_{{c'}\in\mathcal{C}}k'_{c'}\eta'_{c'}.
\end{align} 

{Having defined $Y_c$ in \eqref{Y_c_sojourn}, we proceed to derive the MGF $\bar{M}_{Y_c}(s)=\mathbb{E}[e^{sY_c}]$.} Let $ K_{c'},\bar K_{c'},$ and $K'_{c'}$ denote the random variables representing the numbers of occurrences of random variables $\eta_{c'}$, $\bar \eta_{c'}$, and $\eta'_{c'}$, respectively. Then, using \eqref{Y_c_sojourn}, the MGF of $Y_c$ is calculated as  
\begin{align}\label{mgfequ}
&\bar{M}_{Y_c}(s)=\mathbb{E}[e^{sY_c}]=\\\nonumber
&\mathbb{E}\Big[\mathbb{E}[e^{sY_c}\mid (K_1,\cdots,K_C,\bar K_1,\cdots,\bar K_C,K'_1,\cdots,K'_C)=(k_1,\cdots,k_C,\bar k_1,\cdots,\bar k_C,k'_1,\cdots,k'_C)]\Big]\\\nonumber
&=\sum_{k_1,\cdots,k_C,\bar k_1,\cdots,\bar k_C,k'_1,\cdots,k'_C}\!\!\!\!\!\!\mathbb{E}\big[e^{s(\sum_{{c'}\in\mathcal{C}}k_{c'}\eta_{c'}+\sum_{{c'}\in\mathcal{C}}\bar k_{c'}\bar\eta_{c'}+\sum_{{c'}\in\mathcal{C}}k'_{c'}\eta'_{c'})}\big]\\\nonumber
&\hspace{6mm}\mathrm{Pr}\bigg((K_1,\cdots,K_C,\bar K_1,\cdots,\bar K_C,K'_1,\cdots,K'_C)=(k_1,\cdots,k_C,\bar k_1,\cdots,\bar k_C,k'_1,\cdots,k'_C)\bigg)\\\nonumber
&\stackrel{(a)}{=}\sum_{k_1,\cdots,k_C,\bar k_1,\cdots,\bar k_C,k'_1,\cdots,k'_C}\prod_{c'=1}^{C}\mathbb{E}[e^{s\eta_{c'}}]^{k_{c'}}\prod_{{c'}=1}^{C}\mathbb{E}[e^{s\bar\eta_{c'}}]^{\bar k_{c'}}\prod_{{c'}=1}^{C}\mathbb{E}[e^{s\eta'_{c'}}]^{k'_{c'}}\\\nonumber
&\hspace{6mm}\prod_{{c'}=1}^{C}p_{c'}^{k_{c'}}\prod_{{c'}=1}^{C}\bar p_{c'}^{\bar k_{c'}}\prod_{{c'}=1}^{C}{p'_{c'}}^{k'_{c'}}
Q(k_1,\cdots,k_C,\bar k_1,\cdots,\bar k_C,k'_1,\cdots,k'_C),
\end{align}
where equality $ (a) $ follows because i) random variables $\eta_{c'}$, $\bar \eta_{c'}$, and $\eta'_{c'}$ for all ${c'\in\mathcal{C}}$ are independent, and ii) because of the independence of paths, $ \mathrm{Pr}\bigg(K_1,\cdots,K_C,\bar K_1,\cdots,\bar K_C,K'_1,\cdots,K'_C)=(k_1,\cdots,k_C,\bar k_1,\cdots,\bar k_C,k'_1,\cdots,k'_C\bigg) $ is equal to the summation of the probabilities of all the possible paths {corresponding to the occurrence combination} $ (k_1,\cdots,k_C,\bar k_1,\cdots,\bar k_C,k'_1,\cdots,k'_C) $, which is given by the term $\prod_{{c'}=1}^{C}p_{c'}^{k_{c'}}\prod_{{c'}=1}^{C}\bar p_{c'}^{\bar k_{c'}}\prod_{{c'}=1}^{C}{p'_{c'}}^{k'_{c'}}
Q(k_1,\cdots,k_C,\bar k_1,\cdots,\bar k_C,k'_1,\cdots,k'_C)$, where $Q(k_1,\cdots,k_C,\bar k_1,\cdots,\bar k_C,k'_1,\cdots,k'_C)$  is the number of paths  with the occurrence combination $ (k_1,\cdots,k_C,\bar k_1,\cdots,\bar k_C,k'_1,\cdots,k'_C) $. 

In the following remark, the values of $\mathbb{E}[e^{s\eta_{c'}}],~\mathbb{E}[e^{s\bar\eta_{c'}}]$, and $ \mathbb{E}[e^{s\eta_{c'}'}] $ for all ${c'}\in\mathcal{C}$ are given.
\begin{rem}\label{rem01}
By using the PDFs presented in  Lemma \ref{lemmfaap}, we have 
\begin{align}
&\mathbb{E}[e^{s\eta_{c'}}]=\dfrac{\lambda}{\lambda-s},\\&\nonumber
\mathbb{E}[e^{s\bar\eta_{c'}}]=\dfrac{M_S(s-\lambda_{c'})}{L_{\lambda_{c'}}},\\&\nonumber
\mathbb{E}[e^{s\eta'_{c'}}]=\dfrac{\lambda_{c'}(1-M_S(s-\lambda_{c'}))}{(\lambda_{c'}-s)(1-L_{\lambda_{c'}})}.
\end{align}
\end{rem}

What remains in deriving $ \bar{M}_{Y_c}(s) $ given by the right-hand side of equality $(a)$ of \eqref{mgfequ} are: i) the calculation of $Q(k_1,\cdots,k_C,\bar k_1,\cdots,\bar k_C,k'_1,\cdots,k'_C)$, i.e., the number of paths with the occurrence combination $ (k_1,\cdots,k_C,\bar k_1,\cdots,\bar k_C,k'_1,\cdots,k'_C)$, and ii) calculation of the summation over the different occurrence combinations. While a direct analytical solution seems difficult, we cope with this challenge through the following lemma, providing an effective tool for the remaining calculation.

\begin{lemm}\label{lembro}
Consider a directed graph $G=(\mathcal{V},\mathcal{E})$ consisting of  a set $\mathcal{V}$ of $V$ nodes, a set $\mathcal{E}$ of $E$ edges,  an algebraic
label $ e_{v'\rightarrow \bar v} $ on each edge $e\in\mathcal{E}$ from node $v'$ to $\bar v$, and a node $u\in\mathcal{V}$ with no incoming
edges. 
Let the transfer function $H(v)$ denote the weighted sum over all paths from $u$ to $v$ where the weight of each path is the product of its edge labels.
Then, the transfer functions $H(v),~\forall v\in\mathcal{V}$,  are calculated by solving the following system of linear equations:
\begin{align}\label{hv}
\begin{cases}
H(u)=1\\
H(v)=\sum_{v'\in \mathcal{E}}e_{v'\rightarrow v}H(v'),& u\ne v.
\end{cases}
\end{align}
\end{lemm}
\begin{proof}
See \cite[Sect.~6.4]{6Bixio2016}.
\end{proof}

We adopt Lemma~\ref{lembro} to calculate $ \bar{M}_{Y_c}(s) $ as follows. We form the directed graph $G=(\mathcal{V},\mathcal{E})$ by defining its set of nodes $\mathcal{V}$, the directed edges $\mathcal{E}$ of weights $e_{v'\rightarrow \bar v}$, and the transfer functions of each node, $H(v)$, $v\in\mathcal{V}$, so that the right-hand side of equality $(a)$ in \eqref{mgfequ} becomes equal to the transfer function of a node $\bar{v}\in\mathcal{V}$, $H(\bar{v})$. That is, we seek for the relation $\bar{M}_{Y_c}(s)=H(\bar{v})$. The formation of such graph $G$ can readily be understood by perceiving its high similarity to the structure of the semi-Markov chain -- a directed graph -- in Fig.~\ref{Semi-Chain_c}, which was used to characterize $Y_c$ through paths $\{q_0,\ldots,q_0\}$. In order to define the node $u\in\mathcal{V}$ with no incoming edges, we remove the incoming links of $q_0$, thus representing the node $u$, and as a countermeasure, we introduce a virtual node $\bar q_0$ to account for the system state after completing the service of a source $c$ packet. Finally, observing the factors that represent the edge weights on the right-hand side of equality $(a)$ in \eqref{mgfequ}, we depict the directed graph $G$ in Fig.~\ref{Detour_c}. {According to this graph, $ \bar{M}_{Y_c}(s) $ is given by the transfer function from node $q_0$ to node $\bar q_0$, $H(\bar q_0)$. In other words, we have $\bar{M}_{Y_c}(s)=H(\bar q_0)$, which now leads us to solve for $H(\bar q_0)$ based on \eqref{hv}.}
   

The system of linear equations in \eqref{hv} corresponding to the graph depicted in Fig.~\ref{Detour_c} is given as
\begin{align}\label{hv0}
&H( q_0)=1,\\&\nonumber
H( q_{c'})=p_{c'}\mathbb{E}[e^{s\eta_{c'}}]H(q_0)+p'_{c'}\mathbb{E}[e^{s\eta'_{c'}}]H(q_{c'})+ p_{c'}\mathbb{E}[e^{s\eta_{c'}}]H(q'_0),
\forall {c'}\in \mathcal{C},
\\&\nonumber
H( q'_0)=\sum_{{c'}\in\mathcal{C}_{-c}}\bar p_{c'}\mathbb{E}[e^{s\bar \eta_{c'}}]H(q_{c'}),\\&\nonumber
H(\bar q_0)=\bar p_c\mathbb{E}[e^{s\bar \eta_c}]H(q_c).  
\end{align}
By solving the system of linear equations in \eqref{hv0}, $H(\bar q_0)$ 
is given as 
\begin{align}\label{barhq0}
H(\bar q_0)=\dfrac{p_{c}\mathbb{E}[e^{s\eta_{c}}]{\bar p_{c}}\mathbb{E}[e^{s\bar\eta_{c}}]}{\big(1-p'_{c}\mathbb{E}[e^{s\eta'_{c}}]\big)\bigg(1-\sum_{c'\in\mathcal{C}_{-c}}\dfrac{p_{c'}\mathbb{E}[e^{s\eta_{c'}}]{\bar p_{c'}}\mathbb{E}[e^{s\bar\eta_{c'}}]}{1-p'_{c'}\mathbb{E}[e^{s\eta'_{c'}}]}\bigg)}.
\end{align}

Finally, substituting the probabilities $p_{c'}$, $p'_{c'}$, and $\bar p_{c'}$ given in  Lemma~\ref{prob01} and the values of  $\mathbb{E}[e^{s\eta_{c'}}]$, $\mathbb{E}[e^{s\eta'_{c'}}]$, and $\mathbb{E}[e^{s\bar{\eta}_{c'}}]$ given in Remark~\ref{rem01} into \eqref{barhq0} results in the MGF of the interdeparture time of source $c$, $\bar{M}_{Y_c}(s)$, as given
in Proposition~\ref{Lemma2}. 
\end{proof} 

Finally, substituting the MGF of the system time of source $c$ derived in  \eqref{mgfsystemtime} and the MGF of the interdeparture time of source $c$ derived in \eqref{mgfinterde1}  into \eqref{MGFofagegeneral}  results in the MGF of the AoI under the source-aware preemptive policy, $\bar{M}_{\delta_c}(s)$, given in Theorem~\ref{T_source-aware}. In addition, substituting   \eqref{mgfsystemtime} and \eqref{mgfinterde1} into \eqref{MGFpeak} results in the MGF of the peak AoI under the source-aware preemptive policy, $\bar{M}_{A_c}(s)$,  given in Theorem~\ref{T_source-aware}.
\begin{figure}
\centering
\includegraphics[width=.45\linewidth,trim = 0mm 0mm 0mm 0mm,clip]{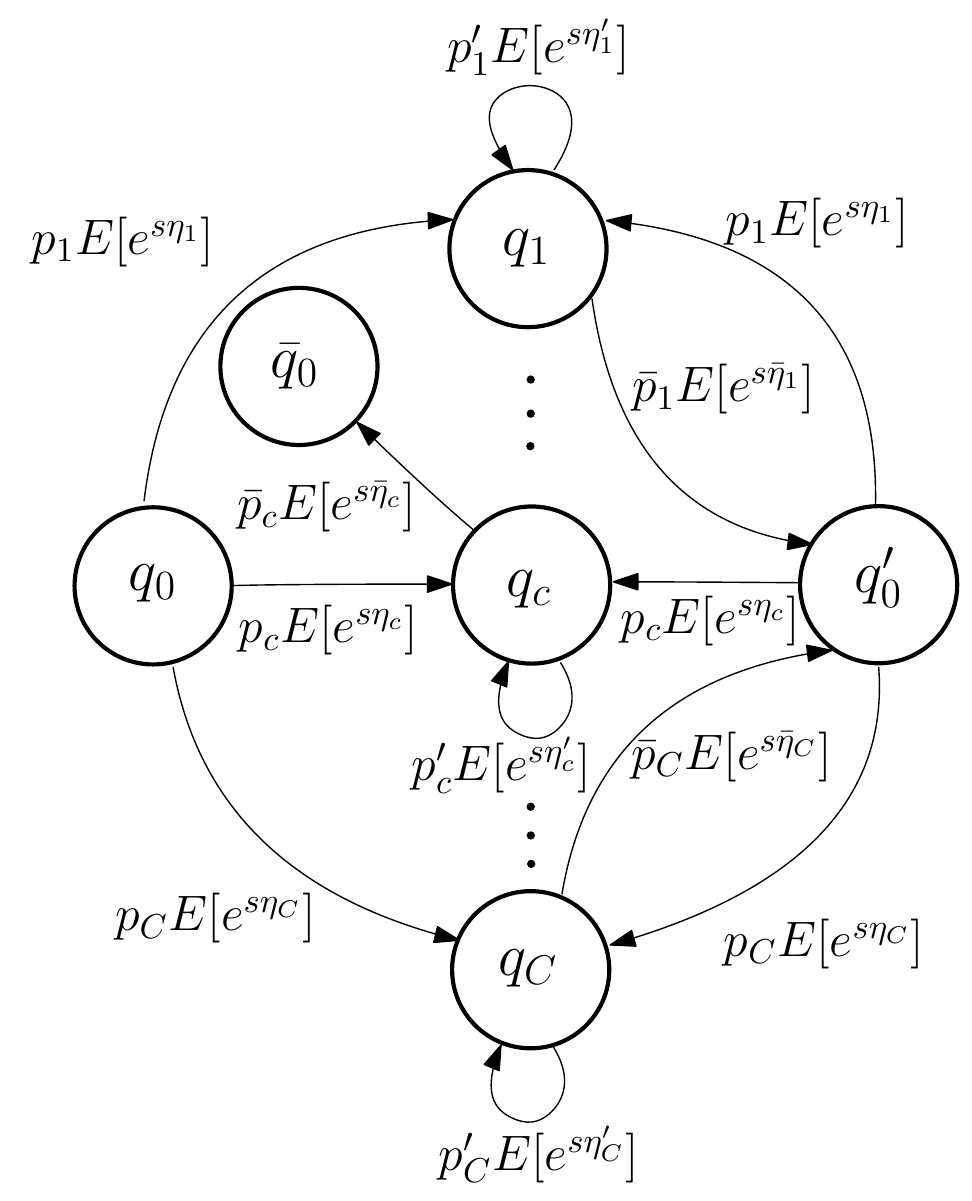}
\caption{The directed graph to calculate the MGF of the interdeparture time under the source-aware preemptive policy.}  
\label{Detour_c}
\vspace{-10mm}
\end{figure}



\subsection{MGFs of AoI and Peak AoI Under the Source-Agnostic Preemptive and Non-Preemptive Policies}

For the source-agnostic preemptive policy, Lemmas~2 and 3 in \cite{8406928} provide the MGFs of the system time of source $c$, $\hat{M}_{T_{c}}(s)$, and
the interdeparture time of source $c$, $\hat{M}_{Y_{c}}$, respectively, which are given as
\begin{align}\label{mgfoftandypreemptive}
&\hat{M}_{T_{c}}(s)= \dfrac{M_S(s-\lambda)}{L_{\lambda}},~~~~
\hat{M}_{Y_{c}}(s)=\dfrac{\lambda_cM_S(s-\lambda)}{\lambda_cM_S(s-\lambda)-s}.
\end{align}
Substituting $\hat{M}_{T_{c}}(s)$ and $\hat{M}_{Y_{c}}(s)$ in \eqref{mgfoftandypreemptive} into \eqref{MGFofagegeneral}  results in the MGF of the AoI under the source-agnostic preemptive policy, $\hat{M}_{\delta_c}(s)$, given in Theorem~\ref{T_preemptive}.  Substituting \eqref{mgfoftandypreemptive} into \eqref{MGFpeak} results in the MGF of the peak AoI of source $c$ under the source-agnostic preemptive policy, $\hat{M}_{A_c}(s)$, given in Theorem~\ref{T_preemptive}.

Under the non-preemptive policy, the system time of a delivered packet is equal to the service time of the packet. Thus, the MGF of the system time of source $c$ under the non-preemptive policy is given by $\tilde{M}_{T_{c}}(s)=M_S(s)$.  Equation (13) in \cite{9500775} provides the MGF of the interdeparture time  of source $c$ under the  non-preemptive policy, $\tilde{M}_{Y_{c}}(s)$, which is given as 
\begin{align}\label{mgfoftandynonpreemptive}
\tilde{M}_{Y_{c}}(s)=\dfrac{\lambda_cM_S(s)}{(\lambda-s)-(\lambda-\lambda_c)M_S(s)}.
\end{align}
Substituting $\tilde{M}_{T_{c}}(s)=M_S(s)$ and $\tilde{M}_{Y_{c}}(s)$ in \eqref{mgfoftandynonpreemptive} into \eqref{MGFofagegeneral}  results in the MGF of the AoI under the non-preemptive policy, $\tilde{M}_{\delta_c}(s)$, given in Theorem~\ref{T_blocking}. 
Substituting $\tilde{M}_{T_{c}}(s)=M_S(s)$ and $\tilde{M}_{Y_{c}}(s)$ in \eqref{mgfoftandynonpreemptive} into \eqref{MGFpeak} results in the MGF of the peak AoI of source $c$ under the non-preemptive policy, $\tilde{M}_{A_c}(s)$, given in Theorem~\ref{T_blocking}. 

%


\section{Numerical Results}\label{Numerical Results}
In this section, we use Corollaries~\ref{agemg11theorem},~\ref{agemg11theorempree},~and~\ref{agemg11theoremblock} to validate the derived results for the average AoI under the source-aware preemptive packet management policy in a two-source system and compare the performance of the three policies in terms of the average AoI and sum average AoI. In addition,  using the MGFs of the AoI derived in Theorems~\ref{T_source-aware},~\ref{T_preemptive},~and~\ref{T_blocking}, we investigate the standard deviation of the AoI to assess the variation of the AoI around the mean. 

We investigate two service time distributions: i) gamma distribution and ii) Pareto distribution. 
\begin{itemize}
    \item The PDF of  a random variable $S$ following a  gamma distribution is defined as
$
{f_S(t)=
\dfrac{\beta^{\kappa} t^{\kappa-1}\exp(-\beta t)}{\Gamma(\kappa)},~t>0,}$ for parameters $ \kappa>0$ and $\beta>0,
$ where $\Gamma(\kappa)$ is the gamma function at $\kappa$. The  service rate is $\mu=1/\mathbb{E}[S]={\beta}/{\kappa}$. 
\item The PDF of a random variable $S$ following a Pareto distribution is defined as
$
{f_S(t)=\dfrac{\alpha {\omega}^\alpha}{t^{\alpha+1}},\,\,\, \text{for}\,\,\,t\in[\omega,\infty]}$ and parameters $ \omega>0$ and $\alpha>1$. The  service rate is  ${\mu=\dfrac{\alpha-1}{\alpha \omega}}$.
\end{itemize}
In all the figures, we have $\lambda=\lambda_1 + \lambda_2= 1$. Next, we investigate the contours of achievable average AoI pairs, standard deviation of the AoI, and the sum average AoI under each policy.




\subsection{Contours of Achievable Average AoI Pairs}
$\text{Fig.\ \ref{Gamma_C}}$ illustrates  the  contours of achievable average AoI pairs $(\Delta_1,~\Delta_2)$ for the proposed source-aware preemptive packet management policy, the source-agnostic preemptive policy, and the non-preemptive policy under the gamma distribution with service rate $\mu=1$ for the  parameters $\kappa=\beta=0.5$, $\kappa=\beta=1.7$, and $\kappa=\beta=3$. Note that for a fixed service rate, increasing $\beta$ makes the gamma distribution to have a lighter tail.    
For the parameters  $\kappa=\beta=0.5$,  the source-agnostic preemptive policy outperforms the others and the non-preemptive is the worst policy (Fig.~\ref{Gamma_1}); for the parameters  $\kappa=\beta=1.7$,  the source-aware preemptive policy outperforms the others and the non-preemptive is the worst policy (Fig.~\ref{Gamma_4}); and  for the parameters  $\kappa=\beta=3$,   the non-preemptive policy outperforms the others and the source-agnostic preemptive policy is the worst one (Fig.~\ref{Gamma_10}). 

Fig.~\ref{Pareto_C} 
illustrates  the  contours of achievable average AoI pairs ${(\Delta_1,~\Delta_2)}$ for the packet management policies
under the Pareto distribution with $\mu=10$ for the sets of parameters ${(\alpha=2.4,~\omega=0.0583)},~{(\alpha=2.7,~\omega=0.630)}$, and ${(\alpha=4,~\omega=0.750)}$. 
Note that for a fixed service rate, increasing $\alpha$ makes the Pareto distribution to have a lighter tail.
%
%
%
 Similar to the observations made for the gamma distribution, for the parameters  ${(\alpha=2.4,~\omega=0.0583)}$, 
the source-agnostic preemptive policy outperforms the others and the non-preemptive policy is the worst one (Fig.~\ref{Pareto_onetenth_n}); for the parameters  ${(\alpha=2.7,~\omega=0.630)}$, 
the source-aware preemptive policy outperforms the others and the non-preemptive policy is the worst one (Fig.~\ref{Pareto_1_n}); and  for the parameters  ${(\alpha=4,~\omega=0.750)}$, 
the non-preemptive policy outperforms the others and the source-agnostic preemptive policy is the worst one (Fig.~\ref{Pareto_2_n}).


{Figs.\ \ref{Gamma_C}} and~\ref{Pareto_C} 
show  that for a fixed mean service time and the set of parameters that make the tail of the distribution heavy enough, the source-agnostic preemptive policy is the best one; and for the parameters that the tail of the distribution is light enough, the non-preemptive policy is the best one.
This is due to the fact that for a fixed mean service time, the heavier the tail, the higher the chance of serving a packet with service time that is substantially longer than the mean service time. In this case, {the preemption enables discarding the packets that would otherwise keep the server inefficiently busy for a long time period and, in turn, enables switching to serve a more fresh packet which has a high chance of experiencing shorter service time.} 
On the other hand, when the tail of the distribution is light enough, it is better to block new arrivals. This is because preemption would cause infrequent updating {due to excessively switching the packet under service so that any packet rarely completes service}.

 

In addition, we can see that the simulated curves for the source-aware preemptive packet management policy matches with the derived expression in Corollary~\ref{agemg11theorem}
(Fig.~\ref{Gamma_1}).

\begin{figure}
\centering
\subfigure[$\kappa=\beta=0.5$]{
\includegraphics[width=0.47\textwidth]{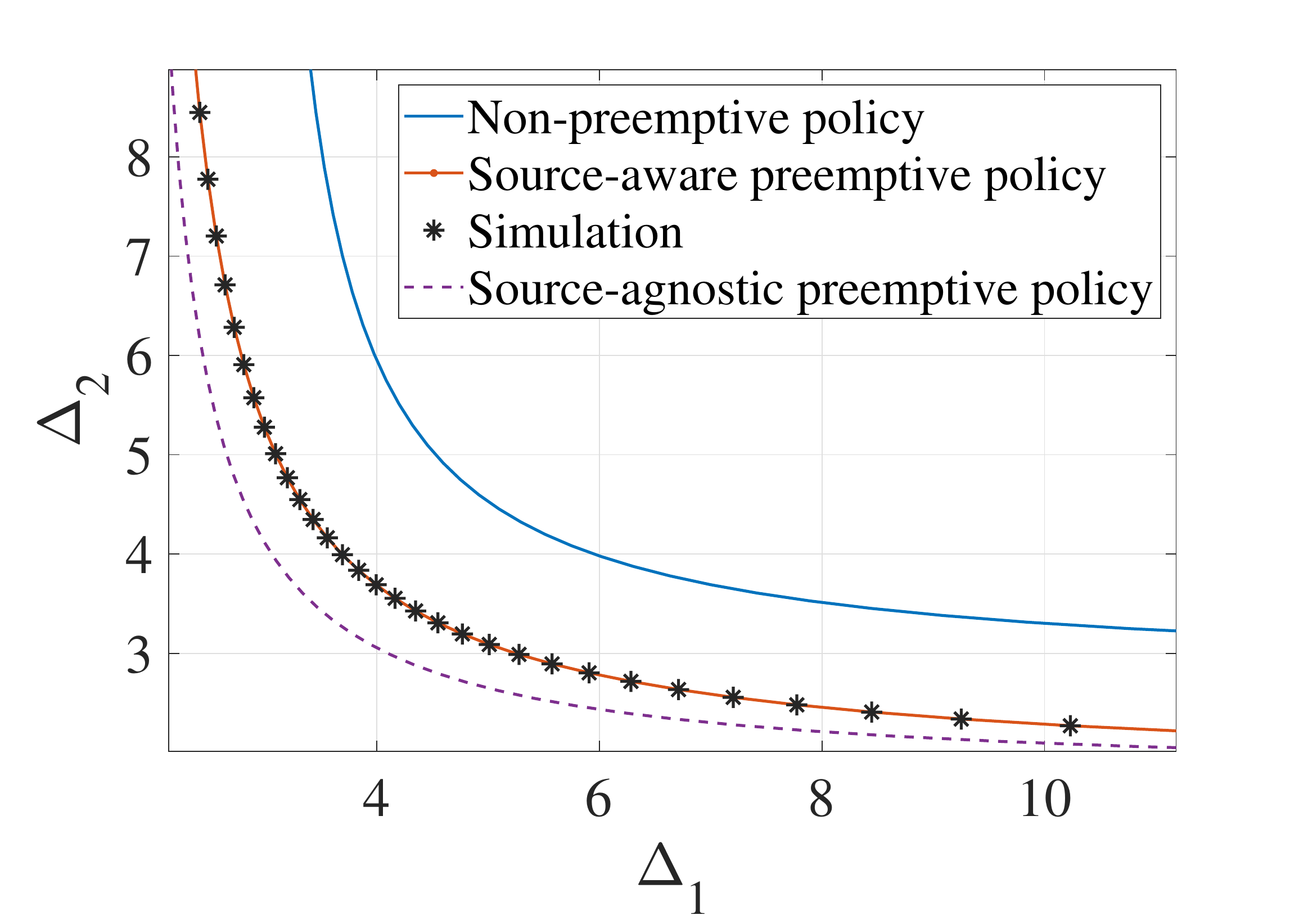}
\label{Gamma_1}
}
\subfigure[$\kappa=\beta=1.7$]
{
\includegraphics[width=0.47\textwidth]{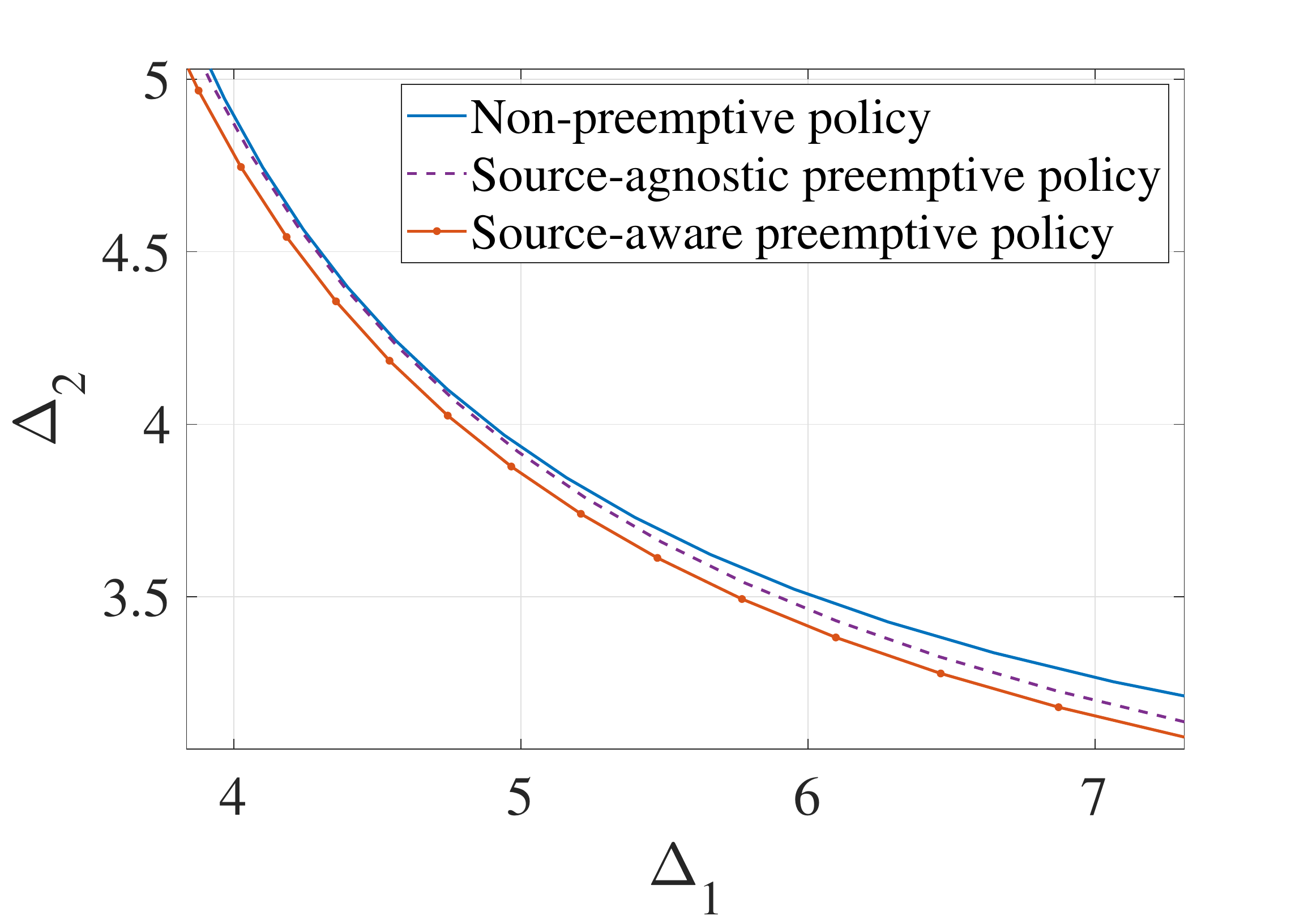}
\label{Gamma_4}
}
\subfigure[$\kappa=\beta=3$]{
\includegraphics[width=0.49\textwidth]{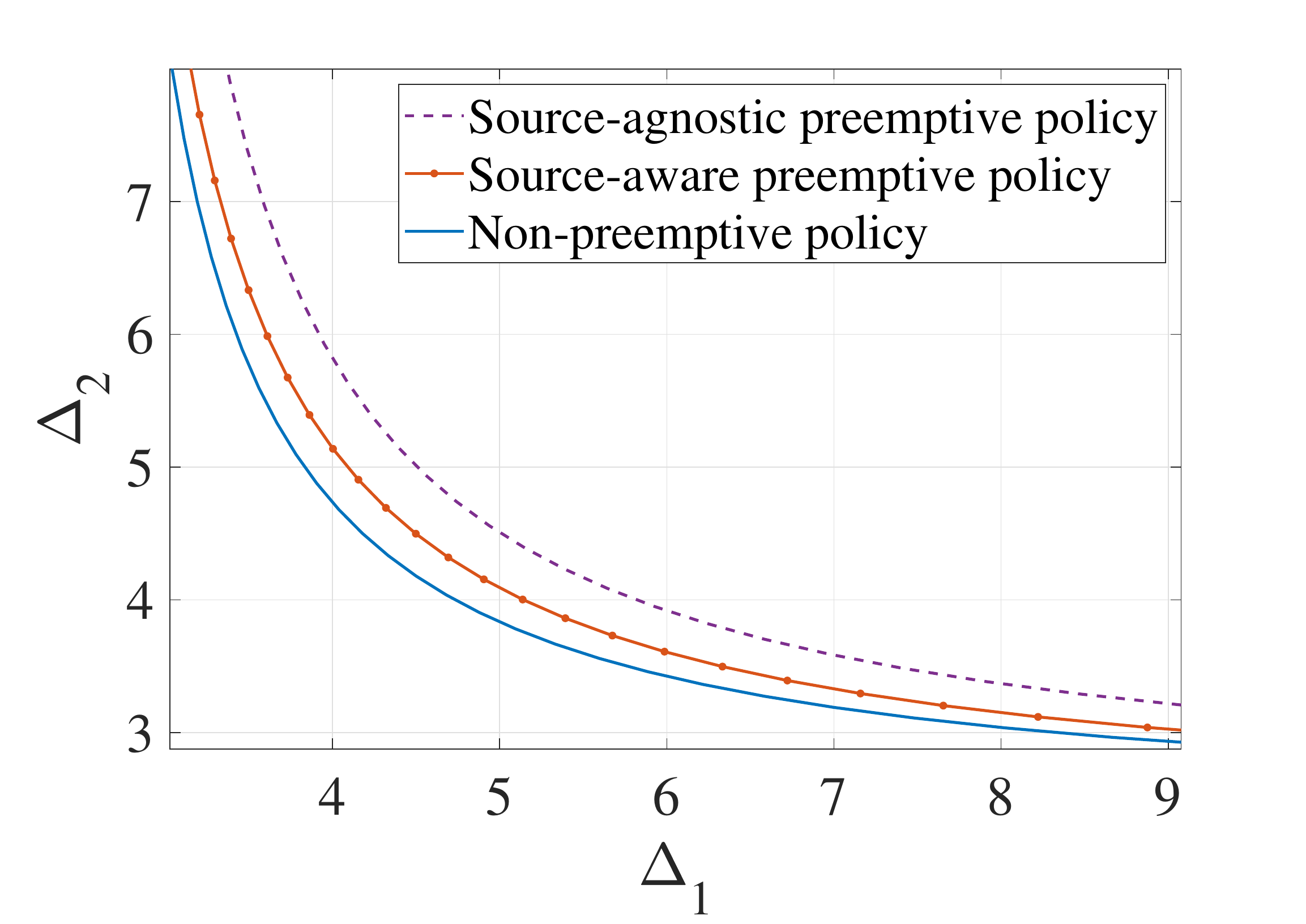}
\label{Gamma_10}
}\vspace{-3mm}
\caption{The contours of achievable average AoI  pairs under the gamma distribution for the different sets of parameters with $\mu=1$.}
\label{Gamma_C}
\vspace{-10mm}
\end{figure}




\begin{figure}
\centering
\subfigure[$\alpha=2.4,\omega=0.0583$ ]{
\includegraphics[width=0.46\textwidth]{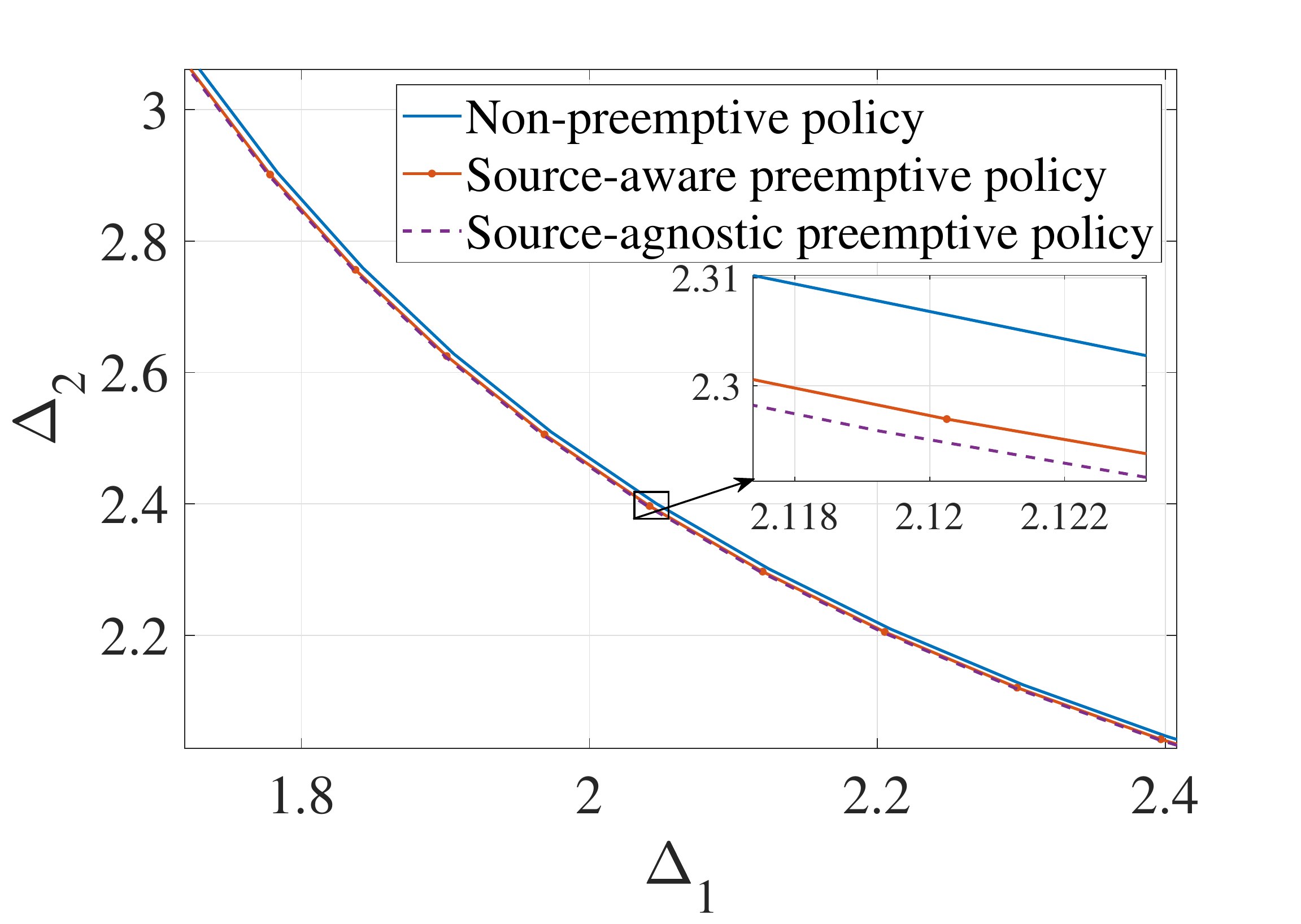}
\label{Pareto_onetenth_n}
}
\subfigure[$\alpha=2.7,\omega=0.630$ ]
{
\includegraphics[width=0.46\textwidth]{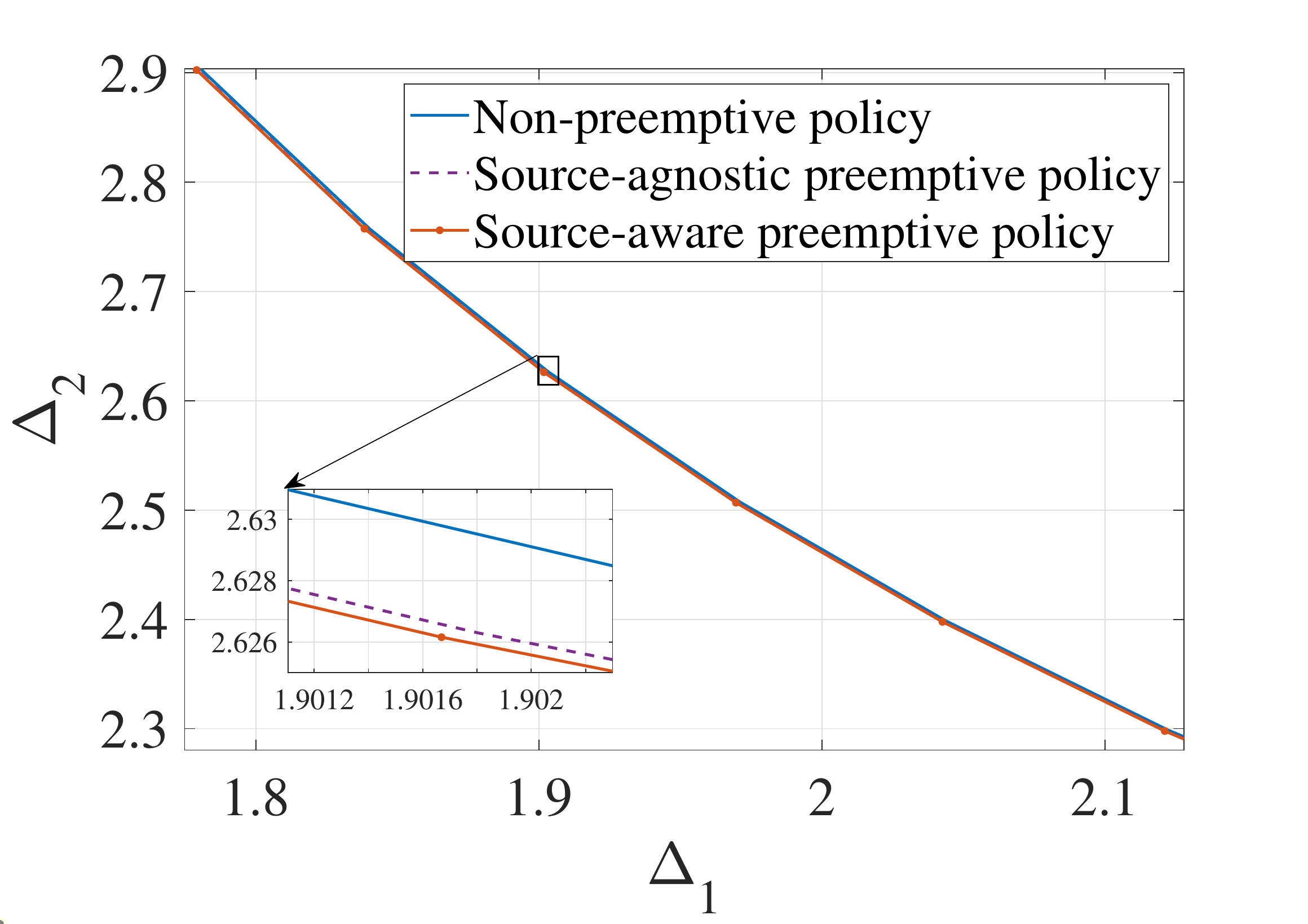}
\label{Pareto_1_n}
}
\subfigure[$\alpha=4,\omega=0.750$]{
\includegraphics[width=0.48\textwidth]{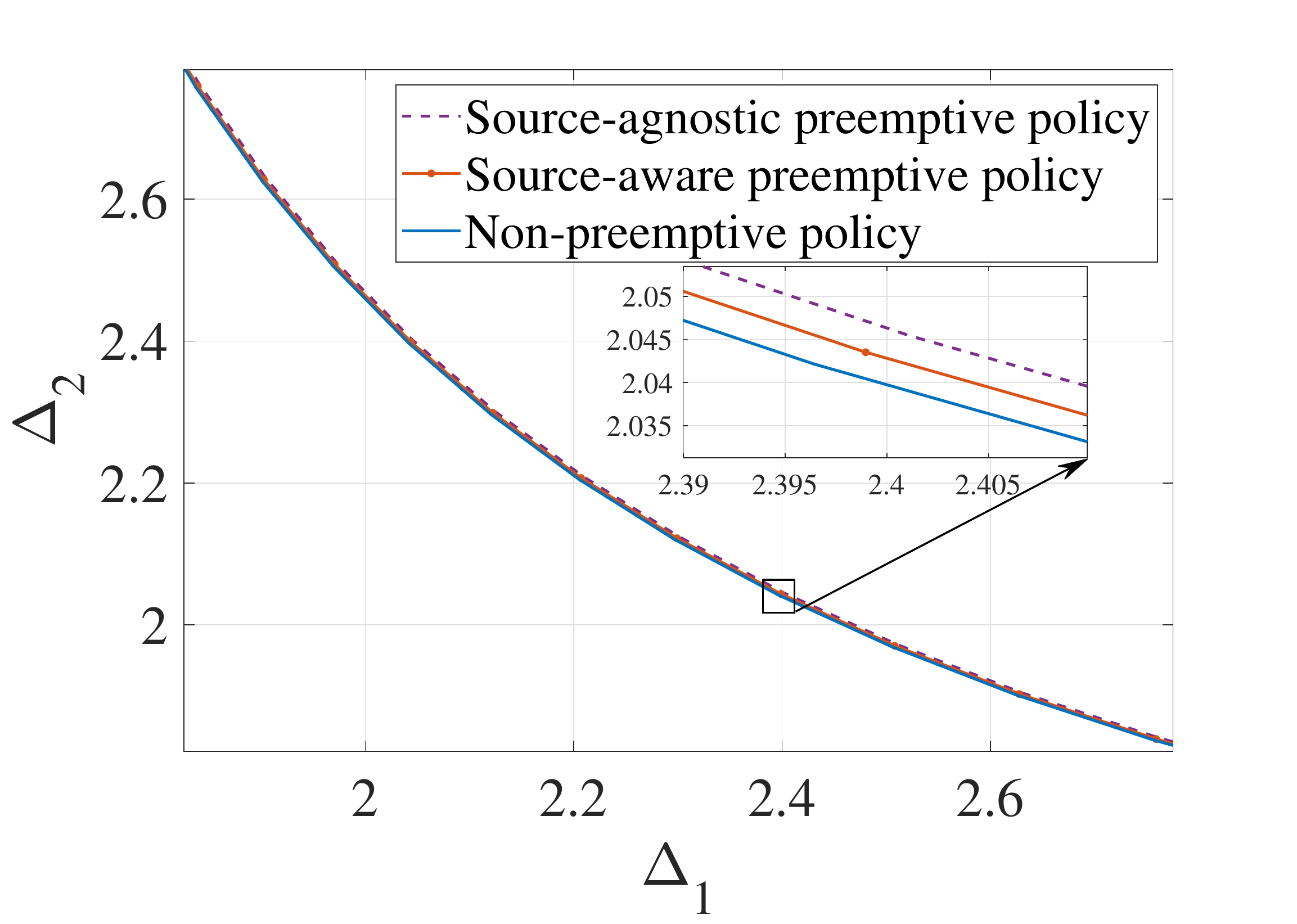}
\label{Pareto_2_n}
}\vspace{-3mm}
\caption{The contours of achievable average AoI  pairs under the Pareto distribution for the different sets of parameters with $\mu=1$.}
\label{Pareto_C}
\vspace{-10mm}
\end{figure}




\subsection{Standard Deviation of the AoI}
Fig. \ref{SD_C}  depicts the average AoI of source 1 and its standard deviation ($\sigma$) as a function of $\lambda_1$ under the gamma distribution\footnote{It is worth noting that since the MGF of the Pareto distribution does not exist,  the standard deviation of the AoI under the Pareto distribution can not be derived.} with parameters $\kappa=2,~\beta=1,~\mu=0.5$ (Fig.~\ref{SDbeta1}) and $\kappa=2,~\beta=4,~\mu=2$ (Fig.~ \ref{SDbeta4}). The standard deviation measures the dispersion of the values of the AoI relative to its mean; {we show this by the curves $\Delta_1+\sigma$ and $\Delta_1-\sigma$}. The figure exemplifies that the standard deviation of the AoI might have a large value {even though the average AoI remains low. For example, while the average AoI performance of the non-preemptive policy is inferior to the other two policies for smaller arrival rates (around ${\lambda_1<0.62}$), the non-preemptive policy results in the least variation of the AoI around its mean for all arrival rates. This demonstrates that the average AoI does not provide complete characterization for the information freshness and thus, higher moments of the AoI need to taken into account when designing and evaluating a reliable status update system. Indeed, besides the requirement of a low average AoI value, maintaining low variation of the AoI values is crucial for time-critical applications. }



\begin{figure}
\centering
\subfigure[$\kappa=2$ and $\beta=1$.]{
\includegraphics[width=0.58\textwidth]{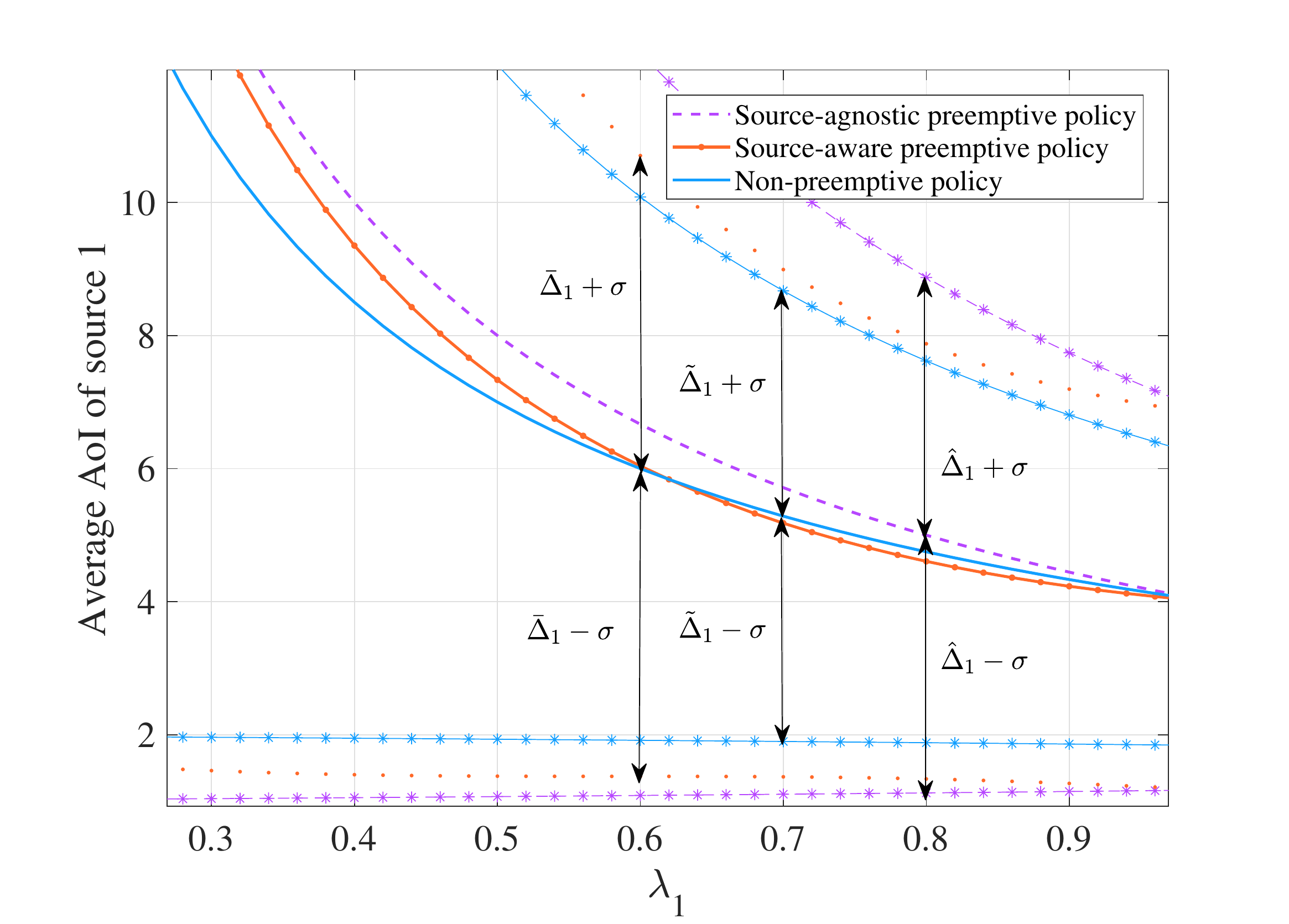}
\label{SDbeta1}
}\vspace{-2mm}
\subfigure[$\kappa=2$ and $\beta=4$.]
{
\includegraphics[width=0.56\textwidth]{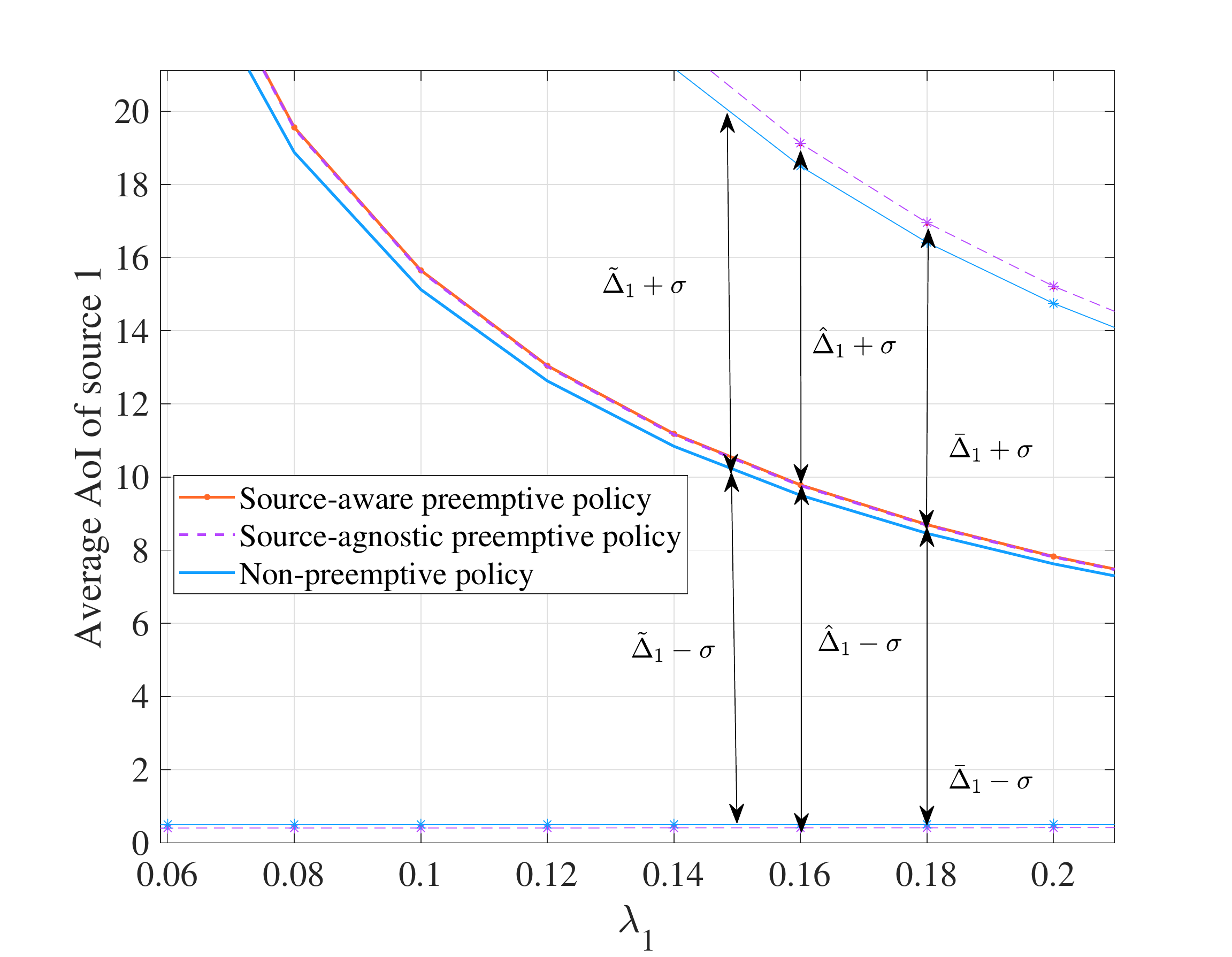}
\label{SDbeta4}
}\vspace{-3mm}
\caption{The average AoI of source 1 and its standard deviation ($\sigma$)
as a function of $\lambda_1$ under  the gamma  distribution.}
\label{SD_C}
\vspace{-10mm}
\end{figure}

\subsection{Sum Average AoI}
{Fig.\ \ref{Sumbeta40}} depicts the sum average AoI, ${\Delta_1+\Delta_2}$, under the  gamma distribution as a function of parameter $\kappa$ with $\beta=1$.
{Fig.\ \ref{Sumparetobar10}} depicts the sum average AoI under the  Pareto distribution as a function of parameter $\alpha$ with $\omega=1$. 
{Fig.\ \ref{Sumbeta4}} illustrates the sum average AoI  under the gamma distribution with ${(\kappa=2,\beta=4)}$, and 
{Fig. \ref{Sumparetobar1}} illustrates the sum average AoI under the Pareto distribution with ${(\alpha=2.2,\omega=1)}$. 
{Similar to the observations made above,} Figs.\ \ref{Sum_C2} and \ref{Sum_C1} exemplify that we can find a parametrization of the gamma and Pareto distributed service times so that each of the three policies, in turn, outperforms the others.






\begin{figure}[t]
\centering
\subfigure[Gamma distribution with $\beta=1$.]
{
\includegraphics[width=0.46\textwidth,height=0.35\textwidth]{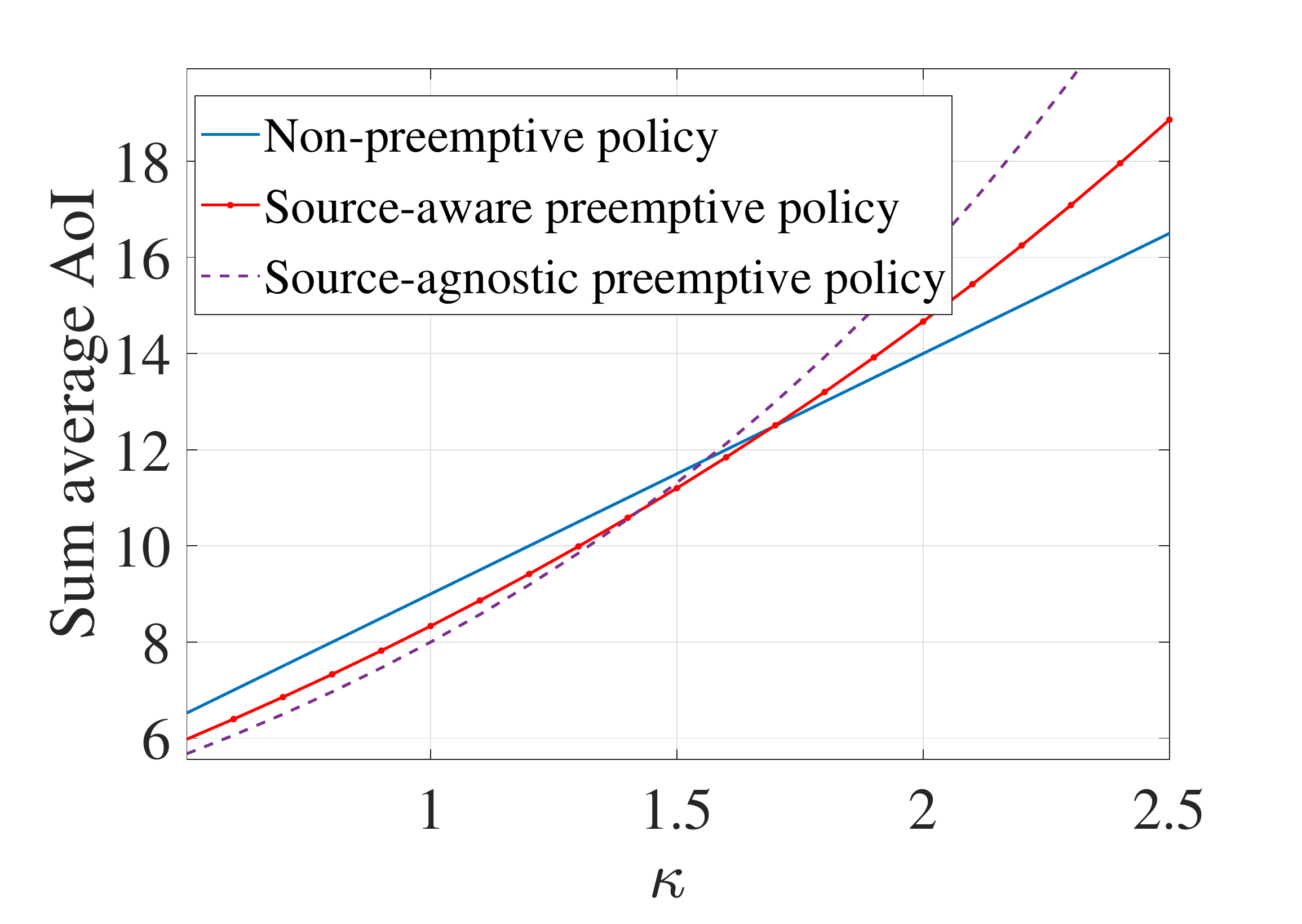}
\label{Sumbeta40}
}
\subfigure[  Pareto distribution  with $\omega=1$.]{
\includegraphics[width=0.46\textwidth,height=0.35\textwidth]{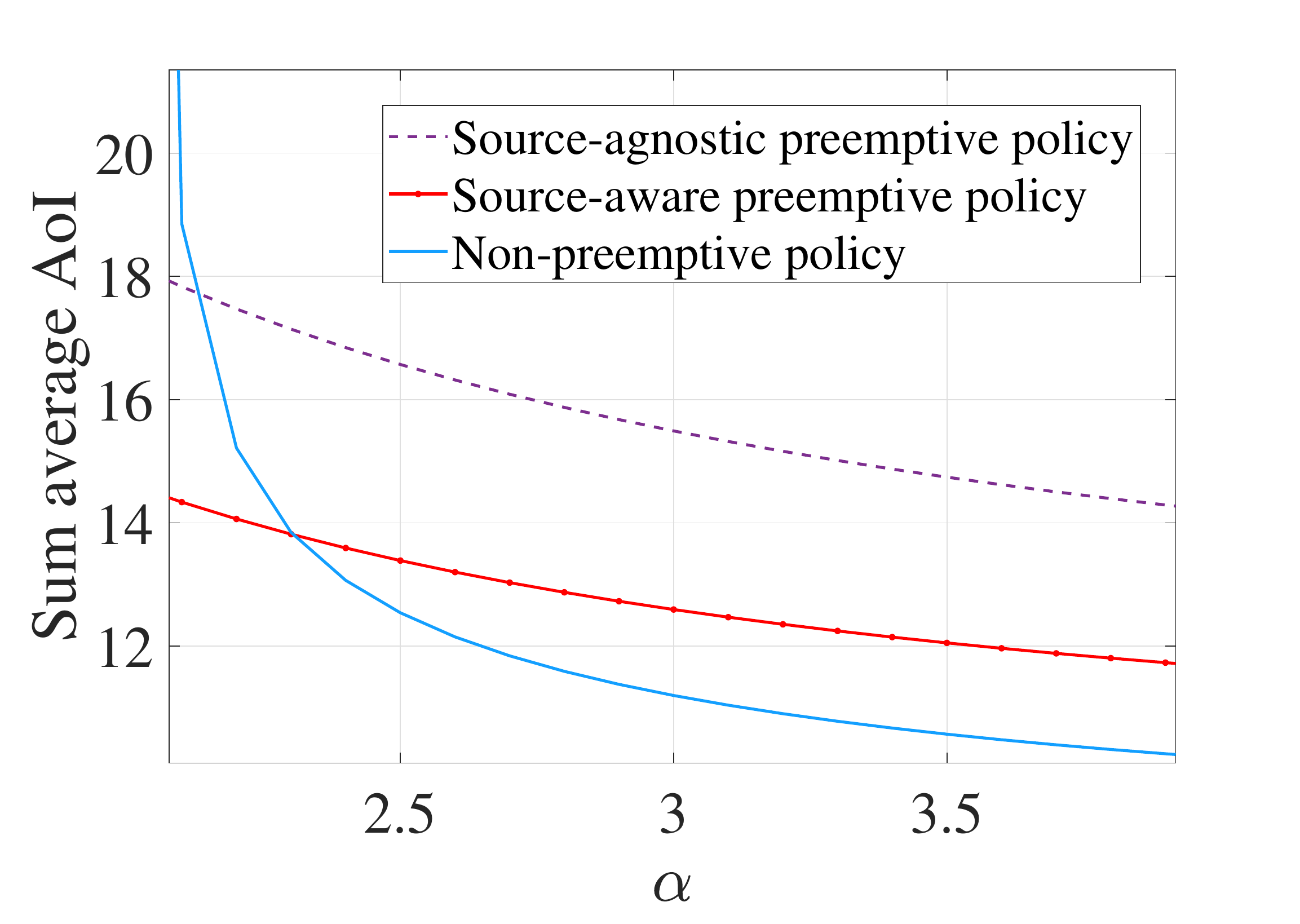}
\label{Sumparetobar10}
}\vspace{-3mm}
\caption{Sum average AoI under gamma and Pareto distributions with $\lambda=0.5$.}
\label{Sum_C2}
\vspace{-10mm}
\end{figure}

\begin{figure}[t]
\centering
\subfigure[Gamma distribution with $(\kappa=2,\beta=4)$.]
{
\includegraphics[width=0.46\textwidth,height=0.35\textwidth]{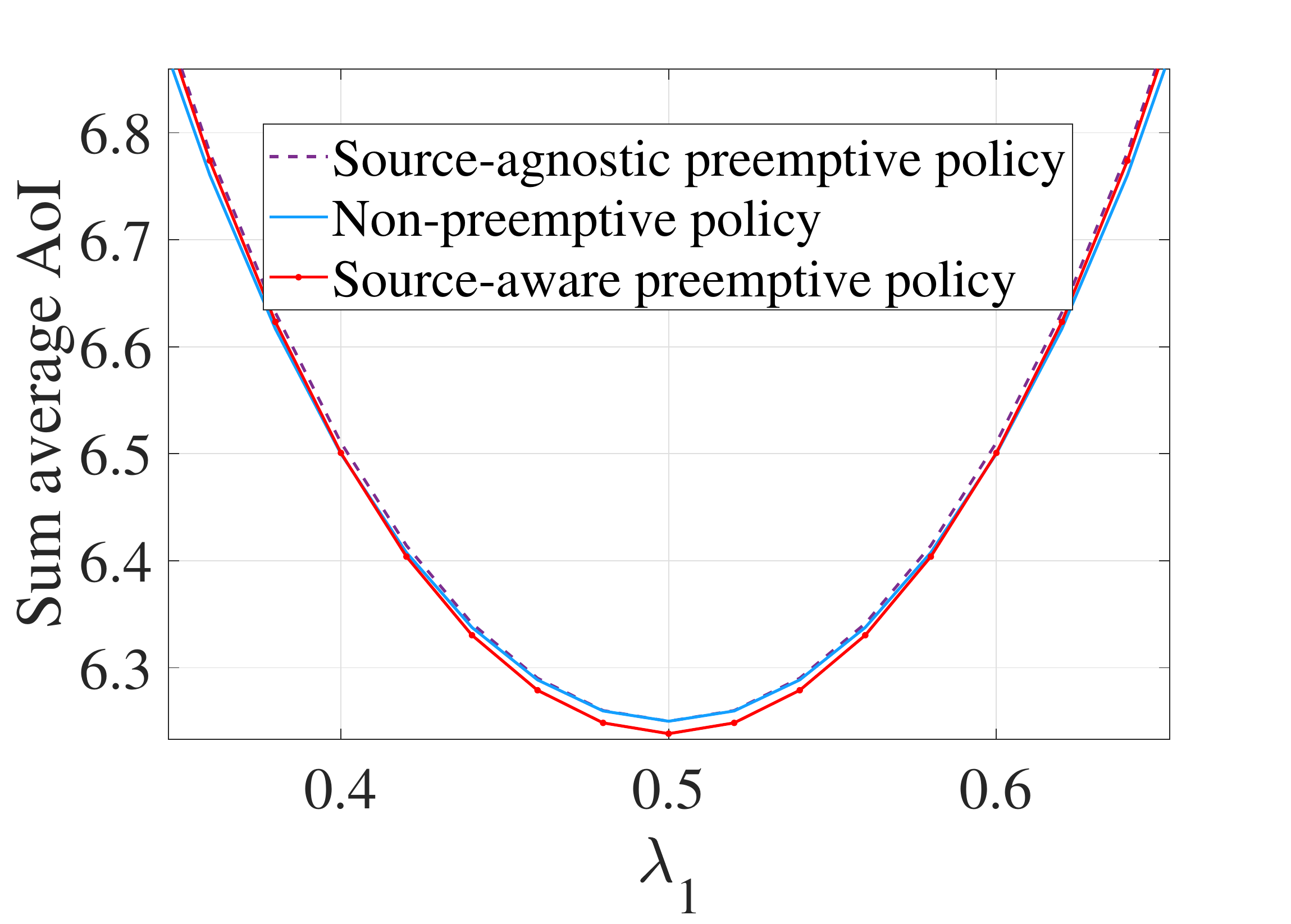}
\label{Sumbeta4}
}
\subfigure[Pareto distribution with $(\alpha=2.2,\omega=1)$.]{
\includegraphics[width=0.46\textwidth,height=0.35\textwidth]{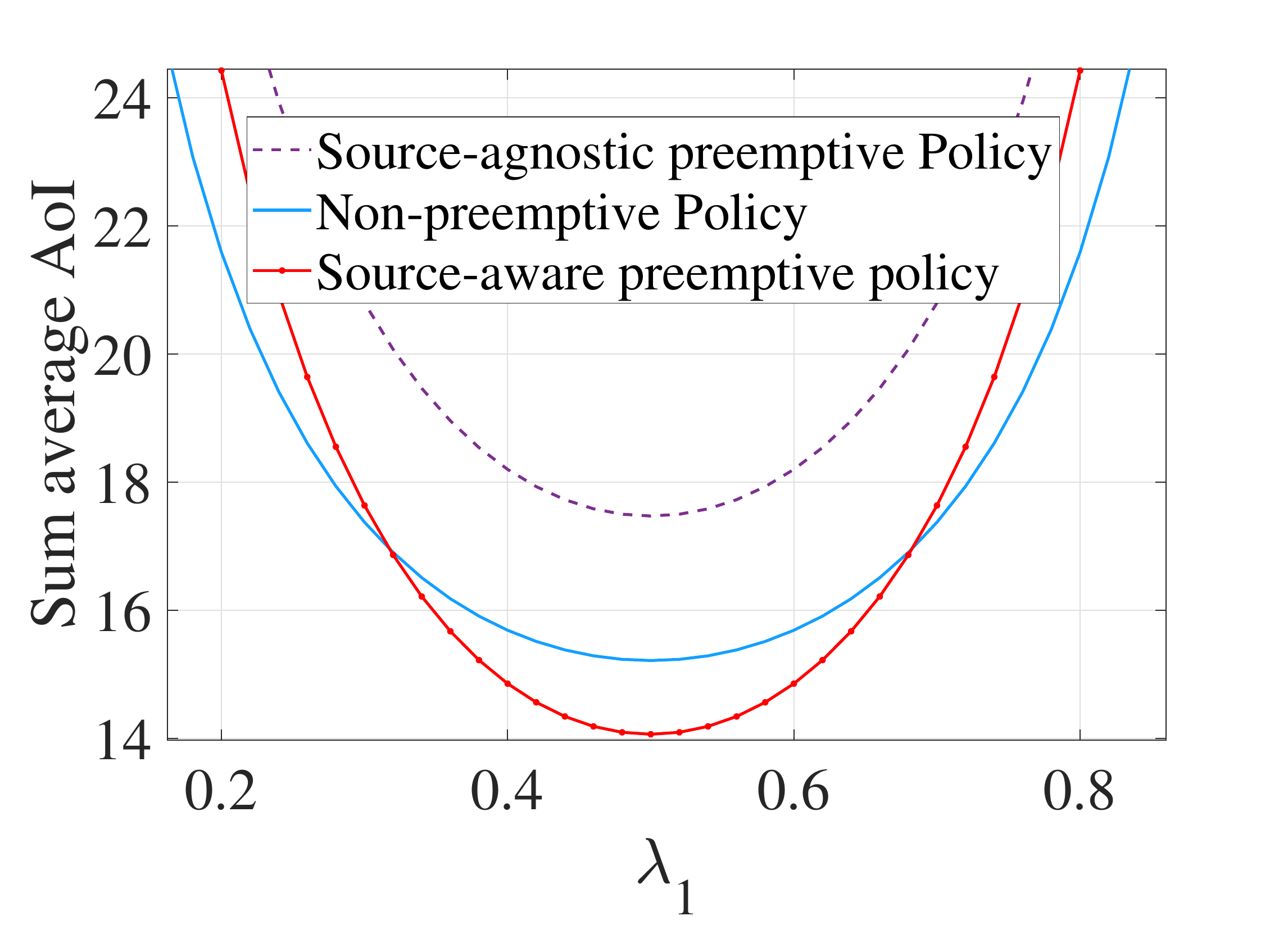}
\label{Sumparetobar1}
}\vspace{-3mm}
\caption{Sum average AoI as a function of $\lambda_1$.}
\label{Sum_C1}
\vspace{-12mm}
\end{figure}

\section{Conclusions}\label{Conclusions}
We derived the MGFs of the AoI and  peak AoI in a multi-source M/G/1/1 queueing model under the proposed source-aware preemptive packet management policy and the source-agnostic preemptive and non-preemptive policies studied earlier. Using the derived MGFs, we derived the average AoI and average peak AoI in a two-source M/G/1/1 queueing system under the three packet management policies.
The numerical results showed that, depending on the system parameters, i.e., the packet arrival rates and the distribution of the service time, each policy can outperform the others. In particular, for a given service rate, when the tail of the service time  distribution is sufficiently heavy, the source-agnostic preemptive policy is the best policy, whereas for a sufficiently light tailed distribution, the non-preemptive policy is the best one. In addition, by visualizing the standard deviation of the AoI, the results demonstrated that the average AoI falls short in thoroughly characterizing the information freshness so that higher moments of the AoI need to be taken into account for the design of reliable status update systems.





\vspace{-5mm}
\bibliographystyle{IEEEtran}
\bibliography{Bibliography}
\end{document}